\documentclass[11pt]{article}
\usepackage{graphicx}
\usepackage{amsthm}
\usepackage{subcaption}
\newtheorem{lemma}{Lemma}

\usepackage{tikz}
\usepackage{arxiv}

\usepackage{graphicx} % Required for inserting images
\usepackage{parskip} %paragraphs no-indent with line break between

\usepackage{float}
\usepackage[font=small]{caption}
\usepackage{algorithm}
\usepackage{algorithmic}
\usepackage{amssymb} 
\usepackage{amsmath}
\usepackage{booktabs}
\usepackage[pdfencoding=auto, psdextra]{hyperref}
\usepackage{amsthm}
\usepackage{thm-restate}
\newtheorem{theorem}{Theorem}

\newtheorem{definition}{Definition}

\newtheorem{proposition}{Proposition}

\usepackage{tikz}
\usetikzlibrary{decorations.pathmorphing,patterns}
\usetikzlibrary{shadows}
\usetikzlibrary{shapes}
\usetikzlibrary{decorations}
\usetikzlibrary{arrows,decorations.markings} 
\usepackage{csvsimple}
\usepackage{color, colortbl}
\usepackage{array}

\usepackage{verbatim}
\usepackage{parskip}% http://ctan.org/pkg/parskip

\usepackage{float} % Required for the 'H' float option
\usepackage{multirow}
\usepackage{caption}
\usepackage{subcaption}
\usepackage{graphicx}
\usepackage{mathtools}
\usepackage{tikz}
\usetikzlibrary{shapes.geometric}
\usepackage{tikz-3dplot}
\usetikzlibrary{arrows,calc,backgrounds}
\usetikzlibrary{decorations.markings}
\usetikzlibrary{backgrounds}
\usepackage{tikz-3dplot}
\usepackage{bbm}
\usetikzlibrary{decorations.pathmorphing,patterns}
\usetikzlibrary{shadows}
\usetikzlibrary{shapes}
\usetikzlibrary{decorations}
\usetikzlibrary{arrows,decorations.markings} 
\usepackage{csvsimple}

\usepackage{mathtools}
\usepackage{tikz}
\usepackage{tikz-3dplot}
\usetikzlibrary{arrows,calc,backgrounds}
\usetikzlibrary{decorations.markings}
\usepackage{bbm}

\pgfmathsetmacro{\r}{3} %
\pgfmathsetmacro{\ra}{2}
\pgfmathsetmacro{\h}{21/8} %
\tdplotsetmaincoords{20}{20}
\usepackage{todonotes} % for comments

\usepackage{amsthm}

\title{Implicit Midpoint Gradient Descent: Fast and Learning rate free convergence for Zero-Sum Games}

\author{
 Gaoqi Xue\\
  Industrial and Systems Engineering\\
  Rensselaer Polytechnic Institute\\
  \texttt{xueg@rpi.edu}\AND James P. Bailey \\
  Industrial and Systems Engineering\\
  Rensselaer Polytechnic Institute\\
  \texttt{bailej6@rpi.edu} 
  }

\date{}
\begin{document}

\maketitle

\begin{abstract}
We study unconstrained bilinear zero-sum games, a fundamental model in online learning, adversarial optimization, and multi-agent decision-making. We introduce the implicit midpoint gradient descent rule, which we derive from continuous-time follow-the-regularized leader dynamics via symplectic integration methods. We prove that implicit midpoint gradient descent inherits several powerful properties from the continuous-time dynamics, including bounded orbits, fast ergodic convergence to Nash equilibria, and learning-rate-independent stability guarantees. This is the first traditional online optimization approach to simultaneously achieve these properties in unconstrained bilinear zero-sum games. Finally, computational experiments demonstrate that the proposed method significantly outperforms the standard methods, optimistic and alternating gradient descent.
% We propose an implicit update method for two-player linear zero-sum games, derived as a structured discretization of continuous-time learning dynamics. The method can be interpreted as an implicit numerical scheme that better preserves the geometric properties of the underlying game dynamics compared to explicit gradient-based updates. In particular, we analyze the induced discrete-time system and show that it inherits stability properties from the continuous-time limit, including bounded trajectories and convergence to the Nash equilibrium under standard regularity assumptions. We further establish that the convergence behavior matches that of continuous-time Follow-the-Regularized-Leader (FTRL) dynamics, while relaxing step-size restrictions typically required for explicit methods. A detailed two-dimensional analysis is provided to characterize the geometric structure of the iterates. Numerical experiments illustrate the theoretical results and demonstrate improved stability relative to classical first-order methods.

\end{abstract}


% 90C47, 91A05, 91A26, 68Q32
% 90C47
% Minimax problems in mathematical programming
% 91A05
% 2-person games
% 91A26
% Rationality and learning in game theory
% 68Q32

% Computational learning theory 

\section{Introduction}

Zero-sum games provide a fundamental framework for modeling competitive interactions in which the objectives of two players are perfectly opposed. The mathematical formulation is the saddle-point problem, where one player minimizes a payoff function while the other maximizes it. The study of zero-sum games is grounded in the minimax theorem, first established by von Neumann (1928) \cite{vonNeumann28}, which guarantees the existence of an equilibrium under mild conditions.

The mathematical framework given by zero-sum games has important applications of modeling adversarial decision-making in economics, machine learning, and control systems. In economics, the pioneering work of von Neumann and Morgenstern \cite{vonNeumannTheoryOfGames} established the foundation of modern game theory through the minimax principle, with important applications in market competition, bargaining, auctions, and decision-making under uncertainty, while further developments by Kuhn and Tucker \cite{kuhn1953}, Blackwell and Girshick \cite{blackwell1954}, and Fudenberg and Tirole \cite{fudenberg1991} strengthened its role in economic analysis and strategic behavior. In modern machine learning, Goodfellow et al. \cite{goodfellow2014generative} introduced Generative Adversarial Networks (GANs), one of the most influential applications of zero-sum games, where a generator and a discriminator are trained through a minimax optimization process; this adversarial formulation has become central to deep generative modeling and has significantly advanced image generation, representation learning, and robust optimization \cite{goodfellow2016}.  In control theory and engineering, Isaacs  \cite{isaacs1965} introduced zero-sum differential games as a foundation for robust control. These ideas have major applications in robotics, missile guidance, cyber-physical systems, and autonomous decision systems.

From a computational perspective, solving bilinear zero-sum games reduces to finding a saddle point of a structured objective; first-order methods are widely used due to their simplicity and scalability. Among the classical methods, the Arrow–Hurwicz and Uzawa methods proposed
in \cite{ArrowHurwicz1958}  are the earliest and most fundamental ones, which simultaneously execute gradient descent on the primal minimizing variables and gradient ascent on the dual maximizing variables. We also have Optimistic Gradient Descent (OGD), which improves stability by incorporating predictive correction terms that reduce the cycling behavior common in standard gradient methods, and the extragradient method, which introduces an additional intermediate step to better handle saddle-point structure and is particularly effective for monotone variational inequalities. These two methods are applied to saddle-point problems \cite{mokhtari2020convergence}, achieving the convergence rate O(1/T). Follow-the-Regularized-Leader (FTRL), developed from online learning theory \cite{shalev2006online}, updates strategies by balancing cumulative past payoffs with regularization and is closely connected to no-regret learning and convergence in repeated games. More recently, Implicit Gradient Descent (IGD), proposed by Essid and Tabak \cite{essid2021implicit}, improves upon explicit methods by incorporating the opponent’s anticipated future response, yielding stronger local convergence and a natural transition toward Newton’s method near saddle points.

Beyond standard first-order methods, Bailey and Piliouras \cite{Bailey19Hamiltonian}  introduced a different perspective by studying continuous-time FTRL dynamics through Hamiltonian systems rather than focusing only on equilibrium computation. They showed that in zero-sum games, continuous-time FTRL satisfies Hamilton’s equations and preserves an invariant energy function, similar to conservative physical systems such as planetary motion. This explains why trajectories often do not converge directly to the Nash equilibrium, but instead move along bounded periodic orbits around it.  Motivated by this, we develop a new method by applying the implicit midpoint discretization scheme to the continuous-time FTRL dynamics, which preserves the Hamiltonian structure more accurately, leading to bounded orbits and more stable long-term behavior while maintaining the geometric properties of the original continuous system.

\subsection{Our Contributions}
We propose a new implicit update method for bilinear zero-sum games, named Implicit Midpoint Gradient Descent (\ref{eq:IMGD}), which can be viewed as a principled formulation of implicit gradient descent. Our method is designed to exactly preserve the energy structure of the underlying continuous-time dynamics, providing a key explanation for its stability. We show that the resulting dynamics achieve convergence rates matching those of continuous-time Follow-the-Regularized-Leader (FTRL), establishing a tight correspondence between discrete-time updates and their continuous limit. Importantly, our method does not restrict the learning rate and remains stable for arbitrarily large step sizes. We demonstrate that the proposed method outperforms standard explicit gradient-based approaches, particularly in regimes where oscillations or divergence occur, highlighting the benefits of energy-preserving implicit discretizations in saddle-point optimization. Table~\ref{tab:learning_methods_comparison} demonstrates the benefits of \ref{eq:IMGD} compared to other baseline methods.

\begin{table}[!ht]
\centering
\footnotesize
\renewcommand{\arraystretch}{1.35}
\setlength{\tabcolsep}{5pt}
\setlength{\extrarowheight}{2pt}
\begin{tabular}{p{0.17\linewidth} p{0.15\linewidth} p{0.37\linewidth} p{0.23\linewidth}}
\toprule
\textbf{Method} & \textbf{Update type} & \textbf{Typical guarantee} & \textbf{Learning-rate restriction} \\
\midrule
Continuous-time gradient/FTRL dynamics
& Continuous-time benchmark
& Bounded orbits \cite{Mertikopoulos2018CyclesAdverserial}; \textbf{ergodic convergence at rate \(O(1/(\eta T))\)}
& No stability restriction; \(\eta\) rescales time \\

\addlinespace[0.35em]
Simultaneous GDA
& Explicit Euler
& Last iterate diverges in bilinear games \cite{Bailey18Divergence}; Can recover \(O(1/\sqrt{T})\) average convergence in constrained settings \cite{cesa2006prediction,bailey2019fast}
& Diminishing steps typically required \\

\addlinespace[0.35em]
Optimistic GDA
& Predictive / one-step memory
& \(O(1/T)\) average convergence in bilinear zero-sum games \cite{daskalakis2019last,mokhtari2020unified} 
& Common bound \(\eta < 1/(2L)\) \\

\addlinespace[0.35em]
Alternating GDA
& Symplectic Euler
& Bounded orbits under suitable steps; \(O(1/T)\) time-average convergence \cite{bailey20261}
& Common bound \(\eta < 2/L\) \\

\addlinespace[0.35em]
Implicit GDA / Proximal point
& Implicit Euler
& Damps rotations and can yield last-iterate convergence in linear settings
& No upper stability restriction  \\

\addlinespace[0.35em]
\textbf{IMGD}
& Implicit midpoint 
& \textbf{Preserves distance to every Nash equilibrium; bounded orbits for all \(\eta\); average residual \(O(1/(\eta T)+1/T)\); two-step residual \(O(1/\eta)\)}
& \textbf{No upper bound on \(\eta>0\)} \\
\bottomrule
\end{tabular}
\caption{Comparison of standard learning dynamics for bilinear zero-sum games with payoff matrix $A$. Our contributions are given in bold. Here, \(L=\|A\|\), and the Nash set is assumed nonempty. Guarantees are representative and depend on the precise setting, norm, and performance metric. }
\label{tab:learning_methods_comparison}
\end{table}

\subsection{Outline of Paper}
Section 2 introduces the two-player zero-sum game framework and reviews the associated continuous-time FTRL dynamics, highlighting their Hamiltonian structure. Section 3 presents our proposed discrete-time update rule and establishes its key properties, including bounded orbits and convergence to the Nash equilibrium. Section 4 provides a detailed two-dimensional analysis to illustrate the dynamics and build intuition for the general case. Section 5 reports numerical experiments comparing our method with standard approaches, demonstrating improved stability and performance. Section 6 concludes with a brief discussion on computational cost and implementation considerations.

\section{Preliminaries}

\subsection{Two-Player Zero-Sum Games}

We consider a two-player zero-sum game with strategy variables 
$x_1 \in X_1 \subset \mathbb{R}^{k_1}$ and 
$x_2 \in X_2 \subset \mathbb{R}^{k_2}$. The game is defined by the saddle-point problem
\begin{equation}\label{eq:minimax}
\max_{x_1 \in X_1} \min_{x_2 \in X_2} u_1(x_1,x_2),
\end{equation}
where player 1's utility function is
\[
u_1(x_1,x_2)
=
x_1^\top A x_2
+
b_1^\top x_1
-
b_2^\top x_2,
\]
with 
$A \in \mathbb{R}^{k_1 \times k_2}$, 
$b_1 \in \mathbb{R}^{k_1}$, and 
$b_2 \in \mathbb{R}^{k_2}$ fixed. Since the game is zero-sum, player 2's utility is
\[
u_2(x_1,x_2)
=
- u_1(x_1,x_2)
=
- x_1^\top A x_2
- b_1^\top x_1
+ b_2^\top x_2.
\] 
A pair of strategies $(x_1^\star, x_2^\star) \in X_1 \times X_2$ is called a Nash equilibrium if
\[
u_1(x_1, x_2^\star)
\le
u_1(x_1^\star, x_2^\star)
\le
u_1(x_1^\star, x_2)
\]
for all $x_1 \in X_1$ and $x_2 \in X_2$. In the unconstrained setting ($X_i=\mathbb{R}^{k_i}$), a Nash equilibrium satisfies the first-order optimality conditions
\[
\nabla_{x_1} u_1(x_1^\star,x_2^\star)=0,
\qquad
\nabla_{x_2} u_2(x_1^\star,x_2^\star)=0.
\]
As a result, the standard measure for the distance to the set of Nash equilibria is the $\ell_2$ norm of the gradient. For most of the results in Sections~3--5, we assume that the game admits a Nash equilibrium, a natural condition. The following proposition provides two sufficient conditions under which this assumption is satisfied for the unconstrained bilinear zero-sum game considered in this paper:

\begin{proposition}[Existence of Nash equilibrium]
Consider the unconstrained zero-sum bilinear game
\(
u_1(x_1,x_2)=x_1^\top A x_2+b_1^\top x_1-b_2^\top x_2.
\)
The game admits a Nash equilibrium if either
\begin{enumerate}
    \item $k_1=k_2$ and $A$ is invertible; or
    \item $b_1=b_2=0$.
\end{enumerate}
\end{proposition}

\begin{proof}
A Nash equilibrium satisfies the first-order conditions
\(
Ax_2^\star+b_1=0,
-A^\top x_1^\star+b_2=0.
\).
If $A$ is invertible, these equations admit the unique solution, $x_2^\star=-A^{-1}b_1,
x_1^\star=(A^\top)^{-1}b_2$.
If $b_1=b_2=0$, the first-order conditions become $Ax_2^\star=0,
-A^\top x_1^\star=0$, 
which are satisfied by $(x_1^\star,x_2^\star)=(0,0)$ and there is at least one Nash equilibrium.
\end{proof}

In this paper, our main results focus on the unconstrained setting, i.e., $X_i=\mathbb{R}^{k_i}$, but much of our analysis extends when $X_i$ is convex and full-dimensional.

\subsection{The Role of the Unconstrained Setting}

In the unconstrained setting, Nash equilibria can be directly computed by solving the linear systems $Ax_2^*+b_1=0$ and $-A^\top x_1^* +b_2=0$. 
Thus, our purpose in studying this setting is not to argue that an iterative learning approach is preferable to direct linear algebra for solving a single game. 
Rather, the unconstrained setting serves as a canonical model for understanding the qualitative behavior of online optimization methods in normal-form (constrained) zero-sum games.

This role is common in the literature; methods such as optimistic gradient descent and alternating gradient descent were first analyzed in unconstrained bilinear games \cite{daskalakis2018limit,mokhtari2020convergence,bailey20261}, where the dynamics of the methods can be isolated, before extending the results to constrained or normal-form settings \cite{daskalakis2019last,nan20251}. 
The reason the model remains informative is that constrained zero-sum dynamics often reduce locally to an unconstrained or affine-constrained system once the active set is identified. 
In particular, for normal-form games, if the support of a Nash equilibrium---the set of strategies with positive weight---is known, then the dynamics restricted to the corresponding face of the simplex evolve on the affine hull of that face. 
In this reduced subspace, the local behavior is determined by the corresponding unconstrained setting. 

As a result, the unconstrained model should be viewed as the basic geometric testbed for local convergence of learning dynamics in the constrained setting. 

\subsection{FTRL Dynamics with Hamiltonian Structure}

To find a Nash equilibrium for \eqref{eq:minimax}, we study dynamics based on the Follow-the-Regularized-Leader (FTRL) algorithm. 
Let \(y_i(t)\) denote the cumulative payoff vector of player \(i\) at time \(t\), and let 
\(h_i : X_i \to \mathbb{R}\) be a strictly convex regularization function. 
In continuous time, the FTRL dynamics take the form
\[
y_i(t)
=
y_i(0)
+
\int_0^t \eta \, \nabla_{x_i} u_i(x(s)) \, ds,
\qquad i = 1,2,
\]
\[
x_i(t)
=
\arg\max_{x_i \in X_i}
\left\{
\langle x_i , y_i(t) \rangle
-
h_i(x_i)
\right\},
\qquad i=1,2.
\]
where \(\eta > 0\) is the learning rate and 
\(x(t) = (x_1(t),x_2(t))\).

% Equivalently, for the zero-sum game considered above,
% \[
% \nabla_{x_1} u_1(x_1,x_2)
% =
% A x_2 + b_1,
% \]
% and
% \[
% \nabla_{x_2} u_2(x_1,x_2)
% =
% A^\top x_1 + b_2.
% \]

% The strategy update is given by

We also define the time-averaged strategies as $
\bar{x}_i(T)
:=
\frac{1}{T}
\int_0^T x_i(t)\,dt$.

For bilinear zero-sum games, the continuous-time FTRL dynamics generate bounded orbits; the trajectories remain confined to a finite region of the state space for all time rather than diverging to infinity. Bailey and Piliouras (2019) \cite{Bailey19Hamiltonian} find that the dynamics on a two-agent zero-sum game are Hamiltonian with an invariant energy function. Therefore,
if the initial strategy is not a Nash equilibrium, the trajectory does not converge to equilibrium but instead stays a fixed positive Bregman distance away and moves around it on a constant-distance orbit.

The notion of a Hamiltonian system originates from classical mechanics and was introduced by William Hamilton to describe the time evolution of physical systems through energy-preserving dynamics.
\begin{definition}[Hamiltonian System]
A dynamical system is called a Hamiltonian system if it can be written in the form
\begin{equation} \label{eq:hamiltonian_system}
\dot{z}(t) = J \nabla H(z(t)),
\end{equation}
where $z(t) \in \mathbb{R}^{2d}$ is the state variable, $H: \mathbb{R}^{2d} \to \mathbb{R}$ is a smooth function called the Hamiltonian (or energy function), and
\[
J = \begin{bmatrix}
0 & I \\
- I & 0
\end{bmatrix}
\]
is the canonical symplectic matrix. 
\end{definition}

Knowing that the continuous-time FTRL dynamics have such nice properties, there exists a simple proof for time-averaged strategies convergent to Nash.

\begin{lemma}\label{lem:ContConverge}
    Suppose $b_i=\vec{0}$, a Nash equilibrium $x^\star$ exists, and that the convex conjugate $h^*_i(y_i)= \max_{x_i \in X_i}\{ \langle x_i, y_i \rangle - h_i(x_i)\}$ is non-negative and has compact level sets. 
    Then $|| \int_{0}^T Ax_2(t)dt / T|| \in O(1/(\eta T))$. 
\end{lemma}

\begin{proof}
    By \cite{Bailey19Hamiltonian}, $H(y(t))= h^*_1(y_1(t)) + h^*_2(y_2(t))$ is constant. Since $h^*_i$, and therefore $H$, is non-negative and has bounded level sets, $||y_i(t)||\leq C$ where $C= \max_{y_i: h^*(y_i) \leq H(y(0))} \lVert y_i \rVert $ .
    Therefore,
    \begin{align*}
        \left\lVert \frac{\int_{0}^T Ax_2(t)dt}{T} \right\rVert &= \frac{ \lVert y_1(T)- y_1(0) \rVert}{\eta T}\leq \frac{C}{\eta T}
    \end{align*}
completing the proof of the lemma. \end{proof}

We remark it is straightforward to extend the result when $b_i\neq \vec{0}$ via the variable substitution introduced in Theorem \ref{thm:bounded}. 
From Lemma \ref{lem:ContConverge}, we know that continuous-time FTRL is highly robust to learning-rate selection, with the learning rate only affecting the speed of convergence but not stability.  Also, the Hamiltonian interpretation of zero-sum game dynamics provides a powerful perspective for algorithm design and analysis.
Since continuous-time FTRL can be viewed as a Hamiltonian system, we can
leverage techniques developed for Hamiltonian dynamics to simulate or discretize these flows efficiently. For instance,
symplectic integrators, which are widely used in Hamiltonian mechanics to preserve geometric structure and energy
over long time horizons, can be adapted to design discrete-time FTRL algorithms that inherit the stability and convergence properties of the continuous system. By combining the geometric insights from Hamiltonian dynamics with the structure of FTRL, we can develop novel algorithms that are both robust to step size choices and faithful to
the underlying game geometry, potentially outperforming classical gradient-based methods in zero-sum games.

\subsection{Gradient Based Methods}\label{sec:GradientBased}

In this section, we study several classical gradient-based methods for online optimization in games. In particular, we view these algorithms as discrete approximations of the continuous FTRL learning dynamics introduced in the previous section. Under this perspective, many characteristic properties of classical methods emerge naturally from the underlying continuous-time learning dynamics they approximate.

Note that in the unconstrained setting \(X_i = \mathbb{R}^{k_i}\), choosing the quadratic regularizer $
h_i(x_i)=\frac12\|x_i\|^2$
makes the maximization problem
\[
x_i(t)=\arg\max_{x_i\in \mathbb{R}^{k_i}}\left\{\langle x_i,y_i(t)\rangle-\frac12\|x_i\|^2\right\}
\]
have the explicit solution
\[
x_i(t)=y_i(t)=x_i(0) + \int_{0}^t \eta \nabla_{x_i} u_i(x(s))ds, 
\qquad i = 1,2,
\]
yielding the continuous-time gradient descent dynamics. In this section, we explore common discrete-time gradient-based methods that can be expressed as an approximation of continuous-time gradient descent.

\paragraph{Gradient Descent (GD).}
In a two-player game, gradient descent methods update each player's strategy by moving in the direction that locally improves their own utility. For the minimax optimization problem~\eqref{eq:minimax}, the update rule is:
\[
x_1^{t+1} = x_1^t + \eta \nabla_{x_1} u_1(x_1^t,x_2^t),
\]
\[
x_2^{t+1} = x_2^t + \eta \nabla_{x_2} u_2(x_1^t,x_2^t).
\]

This update rule can be obtained by applying explicit Euler's method to the continuous-time FTRL-induced gradient flow system \cite{Bailey19Hamiltonian}; it connects gradient-based learning in games to numerical discretizations of continuous-time learning dynamics.

Applying  GDA with a constant step size is known to have last-iterate divergence \cite{Bailey18Divergence,bailey2019fast}. 
This divergence is quickly explained when examining GDA as an approximation of continuous-time gradient descent. Since the continuous-time gradient flow admits a Hamiltonian structure, the total system energy is strictly conserved over time. 
Explicit Euler is well-known to artificially inflate the system's energy (distance to the set of Nash equilibria) in Hamiltonian systems \cite{hairer2006geometric}.

This phenomenon also indicates that the rule fails to achieve last-iterate convergence. Nevertheless, $O(1/\sqrt{T})$ time-average convergence can be recovered by employing decaying learning rates in the constrained setting \cite{cesa2006prediction}. More recently, it has been shown that arbitrary fixed learning rates achieve the same results in 2-agent, 2-strategy games \cite{bailey2019fast}.

\paragraph{Alternating Gradient Descent (AGD).}

For alternating gradient descent, agents update their strategies sequentially: 
\[
x_1^{t+1} = x_1^t + \eta \nabla_{x_1} u_1(x_1^t, x_2^t),
\]
\[
x_2^{t+1} = x_2^t + \eta \nabla_{x_2} u_2(x_1^{t+1}, x_2^t).
\]
This coincides exactly with both the symplectic Euler and Verlet-St\"{o}rmer discretizations of the continuous-time dynamics \cite{bailey20261}. Therefore, AGD can be interpreted as a symplectic integrator for the underlying Hamiltonian dynamics.

Under unconstrained bilinear zero-sum settings, ~\cite{Bailey20Regret} showed that AGD exhibits bounded orbits by exactly preserving a modified discrete invariant energy function at each iteration, where the discretization errors incurred in opposing quadrants analytically counterbalance each other. This is natural, since AGD is obtained via symplectic integration, a technique known to approximately preserve the energy from the continuous system \cite{hairer2006geometric}. 

While the conservation of energy implies AGD does not achieve last-iterate convergence, fast time-average convergence is successfully recovered. The empirical time-averaged strategies 
are guaranteed to converge to the exact Nash Equilibrium at an optimal rate of $O(1/T)$, a property that generalizes to complex multi-agent network zero-sum games~\cite{bailey20261}. \cite{nan20251} extends this analysis to normal-form games (strategies are probability vectors) when there is a fully-mixed Nash equilibrium. 
This conservation of energy can be further exploited to characterize the set of Nash equilibria after a finite number of iterations \cite{kim2025parallelizable}. 

Furthermore, AGD offers an advantage regarding its learning rate tolerance~\cite{bailey20261}. It remains structurally stable under a much larger step size than standard anticipatory updates, satisfying the upper bound $\eta < \frac{2}{\|A\|}$; in contrast, the other standard gradient-based method, optimistic gradient descent, requires much smaller learning rates $\eta < \frac{1}{2\|A\|}$.

\paragraph{Implicit Gradient Descent (IGD).}

Implicit gradient methods evaluate the utility gradients at the next iterate rather than the current one. The update rule is given by
\[
x_1^{t+1}
=
x_1^t
+
\eta \nabla_{x_1} u_1(x_1^{t+1},x_2^{t+1}),
\]
\[
x_2^{t+1}
=
x_2^t
+
\eta \nabla_{x_2} u_2(x_1^{t+1},x_2^{t+1}).
\]

This method can be interpreted as an implicit Euler discretization of the continuous-time gradient flow system.
\cite{essid2021implicit} referred to this approach as \emph{implicit twisted gradient descent} and studied it in the context of minimax optimization problems. The implicit update can be written compactly as
\[
z^{t+1}=z^t-\eta JF(z^{t+1}),
\]
with  
\[
z=\begin{pmatrix}x_1\\x_2\end{pmatrix}, \qquad
F(z)=\begin{pmatrix}-\nabla_{x_1} u_1\\ -\nabla_{x_2}u_1\end{pmatrix},  \qquad
J=
\begin{pmatrix}
I & 0\\
0 & -I
\end{pmatrix}.
\]
They performed a first-order Taylor expansion of the future gradient:
\[
F(z^{t+1})
\approx
F(z^t)+H(z^t)(z^{t+1}-z^t),
\]
where \(H\) denotes the Hessian matrix of the payoff function. This leads to the practical update rule
\[
z^{t+1}
=
z^t-\eta (J+\eta H(z^t))^{-1}F(z^t).
\]
The resulting method interpolates between explicit gradient descent and Newton's method: for small learning rates it behaves similarly to standard gradient descent/ascent, while for large learning rates it approaches a Newton-type iteration. Furthermore, Essid et al.~\cite{essid2021implicit} proved that the algorithm converges locally in a neighborhood of a strict local minimax point. To prevent convergence to critical points of the wrong type, they also proposed an adaptive learning-rate strategy together with acceptance conditions tailored to the minimax structure.

\section{Implicit Midpoint Gradient Descent}

We propose a new rule based on the geometric insights from the continuous-time Hamiltonian dynamics.  We call the rule Implicit Midpoint Gradient Descent (IMGD) since it's obtained through the implicit midpoint integration method, which is a one-step, second-order accurate numerical integrator for ordinary differential equations, and the method is a special case of a \emph{symplectic Runge-Kutta} method, meaning that it approximately preserves the geometric structure of Hamiltonian systems. We show this new rule obtains bounded orbits and time-average convergence to the set of Nash equilibria with arbitrary learning rates.

% \subsection{Implicit Midpoint Gradient Descent (IMGD) }  

Our proposed rule, IMGD, can be viewed as an averaging of GDA (explicit Euler) and IGD (implicit Euler).
We express this update rule in three ways: (\ref{eq:IMGD}) is written consistently with the notation used to approximate dynamics as given in Section \ref{sec:GradientBased}. 
(\ref{eqn:StackedIMGD}) is expressed in stacked notation, which will significantly improve the readability of our analysis. 
Finally, the expression in Proposition \ref{prop:Half}, shows that (\ref{eq:IMGD}) can be expressed as two half-steps followed by a normalization factor to maintain bounded orbits. 
This formulation is useful in practical settings as it expresses (\ref{eq:IMGD}) as a sequence of separable updates based only on local information available to each agent, i.e., agents do not have to share information when updating. 

\[
 x_i^{t+1}= x_i^t + \eta \frac{\nabla_{x_i}u_i(x^{t+1}) + \nabla_{x_i}u_i(x^{t})}{2}, \qquad i=1,2 \tag{IMGD}\label{eq:IMGD}
\]

To simplify our analysis, we also define our rule using stacked-variable notation and gradient vector field:
\[
x=
\begin{pmatrix}
x_1\\
x_2
\end{pmatrix}, \qquad
\begin{pmatrix}
\nabla_{x_1}u_1(x)\\
\nabla_{x_2}u_2(x)
\end{pmatrix}=Gx+c,\qquad G=
\begin{pmatrix}
0 & A\\
-A^\top & 0
\end{pmatrix},
\qquad
c=
\begin{pmatrix}
b_1\\
b_2
\end{pmatrix}.
\]
yielding
\begin{equation}\label{eqn:StackedIMGD}
\begin{aligned}
x^{t+1}
&=
x^t
+
\frac{\eta}{2}
\Bigl(
(Gx^{t+1}+c)
+
(Gx^t+c)
\Bigr)\\
&=
\left(I-\frac{\eta}{2}G\right)^{-1}
\left(
\left(I+\frac{\eta}{2}G\right)x^t
+\eta c
\right)\\
&=\Phi_\eta x^t + D
\end{aligned}\tag{IMGD with Stacked Notation}
\end{equation}
where 
$\Phi_\eta = \big(I - \tfrac{\eta}{2} G\big)^{-1} \big(I + \tfrac{\eta}{2} G\big), \quad D =\big(I - \tfrac{\eta}{2} G\big)^{-1} \eta c. $  In the proof of Proposition \ref{prop:Half}, we establish that $\big(I - \tfrac{\eta}{2} G\big)^{-1}$ is well-defined. 

Finally, we present (\ref{eq:IMGD}) as two half steps followed by a normalization factor. 
This presentation shows that (\ref{eq:IMGD}) can be expressed as a separable update rule. 

\begin{proposition}\label{prop:Half}
    Implicit Midpoint Gradient Descent (IMGD) can be expressed as two half steps followed by a shrinkage operation. 
    Formally, 
\begin{align}
    \tilde{x}_1^{t+1/2}
        &= x_1^t + \frac{\eta}{2} A x_{2}^t + \eta b_1
        \notag\\
    \tilde{x}_2^{t+1/2}
        &= x_2^t - \frac{\eta}{2} A^\top x_{1}^t + \eta b_2
        \tag{Half Step with Full Linear Costs}
        \\[1ex]
    \tilde{x}_1^{t+1}
        &= \tilde{x}_1^{t+1/2} + \frac{\eta}{2} A\tilde{x}_{2}^{t+1/2}
        \notag\\
    \tilde{x}_2^{t+1}
        &= \tilde{x}_2^{t+1/2} - \frac{\eta}{2} A^\top \tilde{x}_{1}^{t+1/2}
        \tag{Half Step without Linear Costs}
        \\[1ex]
    x_1^{t+1}
        &= \left(I+\frac{\eta^2}{4}\,AA^\top\right)^{-1}\tilde{x}_1^{t+1}
        \notag\\
    x_2^{t+1}
        &= \left(I+\frac{\eta^2}{4}\,A^\top A\right)^{-1}\tilde{x}_2^{t+1}
        \tag{Shrinkage Operation to Maintain Bounded Orbits}
\end{align}
\end{proposition}
\begin{proof}
    First, we express the updates in Proposition \ref{prop:Half} with stacked notation. 
    \begin{align*}
        \tilde{x}^{t+1/2} &= \left(I+\frac{\eta}{2}G\right) x^t +\eta c\\
        \tilde{x}^{t+1}&= \left(I+\frac{\eta}{2}G\right) \tilde{x}^{t+1/2}\\
        &= \left[\begin{array}{c c}I & \frac{\eta}{2}A\\ -\frac{\eta}{2}A^\top & I\end{array} \right] \tilde{x}^{t+1/2}\\
        x^{t+1} & = \left[\begin{array}{c c}(I + \frac{\eta^2}{4}AA^\top)^{-1} & 0\\ 0& (I + \frac{\eta^2}{4}A^\top A)^{-1} \end{array} \right] \tilde{x}^{t+1}
    \end{align*}
Applying variable substitutions and Schur's complement inversion formula yields 
\begin{align*}
    x^{t+1} & = \left[\begin{array}{c c}(I + \frac{\eta^2}{4}AA^\top)^{-1} & 0\\ 0& (I + \frac{\eta^2}{4}A^\top A)^{-1} \end{array} \right] \left[\begin{array}{c c}I & \frac{\eta}{2}A\\ -\frac{\eta}{2}A^\top & I\end{array} \right]\left( \left(I+\frac{\eta}{2}G\right) x^t +\eta c \right) \\
    &=\left[\begin{array}{c c}I & -\frac{\eta}{2}A\\ \frac{\eta}{2}A^\top & I\end{array} \right]^{-1}\left( \left(I+\frac{\eta}{2}G\right) x^t +\eta c \right) \tag{Schur's Complement Inversion Formula}\\
    &=\left( I - \frac{\eta}{2} G\right)^{-1}\left( \left(I+\frac{\eta}{2}G\right) x^t +\eta c \right)
\end{align*}
yielding (\ref{eqn:StackedIMGD}).
We remark that $(I+\frac{\eta^2}{4}AA^\top)^{-1}$ is well-defined since $I+\frac{\eta^2}{4}AA^\top$ is positive definite. 
Similarly, $(I+\frac{\eta^2}{4}A^\top A)^{-1}$ and, by Schur's Complement Inversion Formula, $(I-\frac{\eta}{2}G)^{-1}$ are well-defined. 
\end{proof}

% \begin{proposition}
%  $I-\frac{\eta}{2}G$ is invertible. 

% \end{proposition}
% \begin{proof}

%  By direct computation,
% \[
% I-\frac{\eta}{2}G
% =
% \begin{pmatrix}
% I & -\frac{\eta}{2}A\\
% \frac{\eta}{2}A^\top & I
% \end{pmatrix}.
% \]

% We apply the Schur complement criterion. Since the upper-left block $I$
% is invertible, the above block matrix is invertible if and only if its
% Schur complement
% \[
% I - \left(\frac{\eta}{2}A^\top\right)\left(-\frac{\eta}{2}A\right)
% =
% I + \frac{\eta^2}{4} A^\top A
% \]
% is invertible.

% Observe that $A^\top A$ is symmetric positive semidefinite. Hence, for
% any nonzero vector $x$,
% \[
% x^\top\!\left(I + \frac{\eta^2}{4} A^\top A\right)x
% =
% \|x\|^2 + \frac{\eta^2}{4}\|Ax\|^2
% > 0.
% \]
% Therefore,
% \[
% I + \frac{\eta^2}{4} A^\top A
% \]
% is symmetric positive definite and thus invertible.

% By the Schur complement argument, this implies that
% \[
% I-\frac{\eta}{2}G
% \]
% is invertible.   
    
% \end{proof}

\subsection{Fixed Points for \ref{eq:IMGD}}

First, we define the fixed points for (\ref{eq:IMGD}), and we show they are equivalent to the set of Nash equilibria, which is an essential property for any learning algorithm. 

\begin{proposition}
A point \(x^\star\) is a fixed point of the  \ref{eq:IMGD} iteration if and only if it is a Nash equilibrium of the underlying linear zero-sum game. Equivalently, $x^\star=\Phi_\eta x^\star+D$ if and only if $Gx^\star+c=0$.
\end{proposition}

\begin{proof}
Suppose \(x^\star\) is a fixed point of the \ref{eq:IMGD} iteration and $x^\star=\Phi_\eta x^\star+D$. Expanding with the definitions of $\Phi_\eta$ and $D$ yields
\begin{align*}
&x^\star
=
\left(I-\frac{\eta}{2}G\right)^{-1}
\left[
\left(I+\frac{\eta}{2}G\right)x^\star+\eta c
\right]\\
\Leftrightarrow\  & 
\left(I-\frac{\eta}{2}G\right)x^\star
=
\left(I+\frac{\eta}{2}G\right)x^\star+\eta c\\
\Leftrightarrow\ & -\eta Gx^\star=\eta c\\
\Leftrightarrow\ & Gx^\star+c=0
\end{align*}
and \(x^\star\) satisfies the Nash equilibrium condition. The reverse direction holds identically.
\end{proof}

\subsection{Obtaining \ref{eq:IMGD} via Symplectic Integration}
Next, we show this update rule is naturally derived by applying the implicit midpoint integration technique, a known symplectic integrator, to continuous-time gradient descent. 
With symplectic integrators approximately preserving geometric structure and energy in Hamiltonian mechanics over long time horizons when learning rates are sufficiently small \cite{Hairer2006EnergyConserve}, we expect \ref{eq:IMGD} to behave similarly to continuous-GD.
Remarkably, in Section \ref{sec:Preserve}, we will uncover that \ref{eq:IMGD} perfectly preserves energy with arbitrary learning rates.

\begin{proposition}\label{prop:Integrator}
    \ref{eq:IMGD} is obtained by applying the implicit-midpoint rule to the continuous-GD dynamics. 
\end{proposition}

\begin{proof}
The implicit midpoint rule is a second-order implicit numerical integrator for ordinary differential equations $\dot{x}=f(x),$
defined by
\[
x^{t+1}
=
x^t
+
\eta\, f\!\left(\frac{x^t+x^{t+1}}{2}\right).
\]

Substituting the gradient from continuous-GD into the implicit midpoint method yields 
\begin{align*}
&x^{t+1} = x^t + \eta \left( G \frac{x^t + x^{t+1}}{2} + c \right)\\
&x^{t+1} = x^t + \frac{\eta}{2} G x^t + \frac{\eta}{2} G x^{t+1} + \eta c\\
\Leftrightarrow \ \  & \left(I - \frac{\eta}{2} G\right) x^{t+1}
= \left(I + \frac{\eta}{2} G\right) x^t + \eta c.\\
\Leftrightarrow \ \ & x^{t+1}
= \left(I - \frac{\eta}{2} G\right)^{-1}
\Big( \left(I + \frac{\eta}{2} G\right) x^t + \eta c \Big).
\end{align*}

\end{proof}

We remark that the resulting discretization is algebraically identical to the trapezoidal rule for approximating dynamical systems in this (bilinear) setting. 
We prefer the interpretation as the implicit midpoint rule as it is a symplectic integrator and therefore preserves key geometric features of the continuous-time learning dynamics, allowing the discrete-time method to inherit important qualitative properties of continuous-time systems.
The trapezoidal rule is not always a symplectic integrator and therefore is less likely to preserve properties in nonlinear systems. 

% \begin{definition}
%  The implicit midpoint method for a general dynamical system $\dot{x}(t) = f(x(t))$ is defined as\begin{equation} \label{eq:midpoint_general}
% x^{k+1} = x^k + \eta\, f\!\left(\frac{x^k + x^{k+1}}{2}\right).\end{equation}
% \end{definition}

% Specifically, in our linear system it satisfies
% \[
% f\!\left(\frac{x^k + x^{k+1}}{2}\right)
% = \frac{1}{2}\big(f(x^k) + f(x^{k+1})\big),
% \]
% so for update rule \eqref{eq:imgd_update} it yields
% \[
% x^{k+1} = x^k + \eta \cdot \frac{1}{2} \big( f(x^k) + f(x^{k+1}) \big).
% \]
% This shows that the implicit midpoint method averages the vector fields (slopes) used by the explicit Euler method, which evaluates $f(x^k)$, and the implicit Euler method, which evaluates $f(x^{k+1})$.

\subsection{Preserving the Orbits from Continuous-GD}\label{sec:Preserve}

In bilinear zero-sum games, the underlying continuous-time dynamics are governed by a skew-symmetric operator, leading to energy-preserving rotational behavior. Consequently, trajectories remain on level sets of a conserved quantity and are inherently bounded. This suggests that a well-designed discretization should inherit this qualitative property.

\ref{eq:IMGD} achieves this exactly: it preserves a discrete energy function, ensuring that the iterates evolve on invariant level sets and remain bounded for all step sizes. In contrast, alternating gradient descent, which is obtained from the symplectic Verlet-St\"{o}rmer integration technique, only approximately preserves this energy \cite{bailey20261,nan20251}; the AGD update introduces perturbations, which can lead to energy growth and divergence of the iterates when learning rates are too large.

\begin{theorem}\label{thm:bounded}
(\ref{eq:IMGD}) preserves the distance to each Nash equilibrium regardless of learning rate selection.  Formally, suppose there exists a Nash equilibrium $x^\star$,we have  $\lVert x^{t+1} - x^* \rVert = \lVert x^{t} - x^*\rVert$ for all $t$.
\end{theorem}

\begin{proof}
    First, observe that $\Phi_\eta$ is orthogonal. Formally, since $G^\top=-G$,
    \begin{align*}
        \Phi_\eta^\top \Phi_\eta &= \left(\left(I - \tfrac{\eta}{2} G\right)^{-1}\left(I + \tfrac{\eta}{2} G\right)\right)^\top \left(I - \tfrac{\eta}{2} G\right)^{-1}\left(I + \tfrac{\eta}{2} G\right)\\
        &= \left(I + \tfrac{\eta}{2} G\right)^\top \left(I - \tfrac{\eta}{2} G\right)^{-\top}\left(I - \tfrac{\eta}{2} G\right)^{-1}\left(I + \tfrac{\eta}{2} G\right)\\
        &= \left(I - \tfrac{\eta}{2} G\right) \left(I + \tfrac{\eta}{2} G\right)^{-1}\left(I - \tfrac{\eta}{2} G\right)^{-1}\left(I + \tfrac{\eta}{2} G\right)\\
        &= \left(I - \tfrac{\eta}{2} G\right) \left(I - \tfrac{\eta}{2} G\right)^{-1}\left(I + \tfrac{\eta}{2} G\right)^{-1}\left(I + \tfrac{\eta}{2} G\right)\tag{since $(I+\frac{\eta}{2}G)$ and $(I-\frac{\eta}{2}G)$ commute}\\
        &= I.
    \end{align*}
Next, consider the shifted sequence $z^t=x^t-x^*$ for all $t$. Since $x^*$ is a Nash equilibrium, $x^*=\Phi_\eta x^* + D$, and 
\begin{align*}
    z^{t+1} = x^{t+1}-x^*= (\Phi_\eta x^t +D) - (\Phi_\eta x^* +D)= \Phi_\eta(x^t-x^*)=\Phi_\eta z^t.
\end{align*}
Finally, since $\Phi_\eta$ is orthogonal, 
\begin{align*}
    \lVert x^{t+1}- x^*\rVert = \lVert z^{t+1} \rVert= \lVert \Phi_\eta z^t\rVert = \sqrt{(z^t)^\top \Phi_\eta^\top \Phi_\eta z^t} =  \sqrt{(z^t)^\top  z^t} =\lVert z^t \rVert = \lVert x^t-x^* \rVert
\end{align*}
yielding the statement of the theorem. 
\end{proof}

% The update rule is
% \[
% x^{t+1} = \Phi_\eta x^t + D.
% \]

%  For  a Nash equilibrium        $ x^\star$ satisfying
% \[
% x^\star = \Phi_\eta x^\star + D.
% \]

% Define the shifted variable
% \[
% z^t := x^t - x^\star.
% \]

% Then
% \[
% z^{t+1}
% = x^{t+1} - x^\star
% = \Phi_\eta x^t + D - (\Phi_\eta x^\star + D)
% = \Phi_\eta (x^t - x^\star)
% = \Phi_\eta z^t.
% \]

% We first show that $\Phi_\eta$ is orthogonal. Recall that
% \[
% \Phi_\eta = \left(I - \tfrac{\eta}{2} G\right)^{-1}\left(I + \tfrac{\eta}{2} G\right),
% \]
% where $G^\top = -G$. A direct computation yields
% \[
% \Phi_\eta^\top \Phi_\eta = I.
% \]
% So we have
% \[
% \|z_{t+1}\|^2
% = z_t^\top \Phi_\eta^\top \Phi_\eta z_t
% = \|z_t\|^2.
% \]

% Thus,
% \[
% \|z^t\| = \|z^0\|, \quad \forall t \ge 0,
% \]
% which implies
% \[
% \|x^t - x^\star\| = \|x^0 - x^\star\|.
% \]

\subsection{Convergence}

We now establish the ergodic convergence of the proposed method.

\begin{theorem}\label{thm:converge}
    Suppose Nash equilibrium exists, let $\{x^t\}$ be generated by the implicit midpoint gradient descent rule and consider the time-average of the strategies, $\bar{x}^T := \frac{1}{T} \sum_{t=0}^{T-1} x^t$. Then $\lVert G\bar{x}^T + c \rVert \leq \frac{\lVert x^0-x^*\rVert}{T}\left( \frac{2}{\eta}+\lVert G\rVert \right)$, i.e., implicit midpoint gradient descent achieves $O\left(\frac{1}{\eta T} + \frac{1}{T} \right)$ time-average convergence to the set of Nash equilibria. 
\end{theorem}

\begin{proof}
Recall from the proof of Proposition \ref{prop:Integrator}, 
\begin{align*}
&\left(I - \frac{\eta}{2} G\right) x^{t+1}
= \left(I + \frac{\eta}{2} G\right) x^t + \eta c\\
\Leftrightarrow\  & x^{t+1} - x^t= \frac{\eta}{2}G(x^{t+1}+x^t) + \eta c.
\end{align*}
Summing this equality over $t=0,\dots,T-1$ yields
\begin{align*}
&\sum_{t=0}^{T-1} (x^{t+1} - x^t)
= \frac{\eta}{2} G \sum_{t=0}^{T-1} (x^{t+1} + x^t)
+ \eta T c\\
\Leftrightarrow \ &x^T - x^0
= \frac{\eta}{2} G \sum_{t=0}^{T-1} (x^{t+1} + x^t)
+ \eta T c.
\end{align*}

Dividing both sides by $\eta T$ gives
\begin{align*}
\frac{x^T - x^0}{\eta T}
&= G  \sum_{t=0}^{T-1} \frac{x^{t+1} + x^t}{2T}
+ c\\
&= G \bar{x}^T + G\frac{x^T-x^0}{2T} +c.
\end{align*}
Finally, by Theorem \ref{thm:bounded}, $\{x^t\}_{t=1}^\infty$ remains equidistant from a Nash equilibrium $x^*$, and $\|x^T - x^0\| = \| (x^T- x^*) - (x^0 - x^*)\| \le \| x^T-x^*\| +\| x^0- x^*\| =  2\| x^0- x^*\|$ implying
\begin{align*}
\lVert G \bar{x}^T + c\rVert
= \left\lVert\frac{x^T - x^0}{\eta T}
- G \cdot \frac{x^T - x^0}{2T}\right\rVert
\leq \frac{\lVert x^T - x^0 \rVert}{\eta T} + \lVert G\rVert   \frac{\lVert x^T - x^0 \rVert }{2T}
\leq \frac{2\| x^0- x^*\|}{\eta T} + \frac{\| x^0- x^*\|\lVert G\rVert }{T}
\end{align*}
completing the proof of the theorem.
\end{proof}

In particular, the learning rate $\eta$ improves the convergence rate in the sense that the leading term scales as $O(1/(\eta T))$ and there is no restriction on the learning rate $\eta > 0$ for stability or convergence of the method. 
However, as $\eta \to \infty$, the rate saturates at $ O\!\left(\frac{1}{T}\right).$
Thus, increasing $\eta$ beyond a certain scale does not further improve the convergence rate in this analysis.

Moreover, when there exist multiple  Nash equilibria, the solution set forms an affine manifold rather than an isolated point. In this setting, convergence can no longer be expected to a unique point. Instead, a natural notion of convergence is to the closest equilibrium within the solution set.

\begin{theorem}
\label{thm:multiple}
    Let ${{\cal X}^\star}= \{x: x : Gx=-c\}$ is the non-empty set of Nash equilibria, $P(x^0)
= \arg\min_{x\in {\cal X}^\star}\|x-x^0\|$.  Then
$\bar{x}^T \to P(x^0)$ as $T\to \infty$.  
\end{theorem}

The proof is identical to the same result for Continuous-time GD in \cite{bailey2026uniqueness}. 
The key to the proof is that, since the strategies remain equidistant to every Nash equilibrium (Theorem \ref{thm:bounded}), both the strategies and the time-average of the strategies are contained in an affine space containing only the equilibrium $P(x^0)$, as depicted in Figure \ref{fig:MultiEquil}. 
Since the time-average of the strategies converges to the set of Nash (Theorem \ref{thm:converge}), and since $P(x^0)$ is the equilibrium contained in an affine space containing the time-average of the strategies, the time-average must converge to $P(x^0)$. 
For completeness, we include a near-verbatim of the proof from \cite{bailey2026uniqueness} in the appendix. 

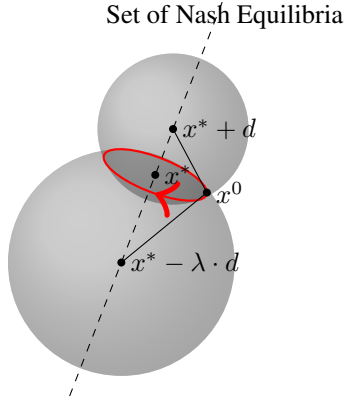
\begin{figure}[!ht]
\centering
% \resizebox{3in}{!}{%
    \begin{tikzpicture}[
    tdplot_main_coords,
    Helpcircle/.style={gray!70!black,
    },scale=0.5
    ]

    \draw[dashed] (0,-4,0) -- (0,8,0) node[below,black] {Set of Nash Equilibria};

    \coordinate (M1) at (0,0,0); 
    \coordinate (M2) at (0,4,0); 
    \coordinate (A) at (0,\h,0);

    \tdplotsetrotatedcoords{90}{90}{0}%
    \draw[thick, Helpcircle, red, tdplot_rotated_coords,decoration={markings, mark=at position 1 with {\pgftransformscale{3}\arrow{>}}},
            postaction={decorate}] (A) circle[radius=sqrt(\r*\r-\h*\h)];

    \begin{scope}[tdplot_screen_coords, on background layer]
    \fill[ball color= gray!20, opacity = 0.25] (M1) circle (\r); 
    \end{scope}

    \begin{scope}[tdplot_screen_coords, on background layer]
    \fill[ball color= gray!20, opacity = 0.25] (M2) circle (\ra); 
    \end{scope}

    \fill[black] ({sqrt(\r*\r-\h*\h)},\h,0) circle[radius=3pt];
    \node[right,black] at  ({sqrt(\r*\r-\h*\h)},\h,0) {$x^0$};

    \fill[black] (M1) circle[radius=3pt];
    \node[right,black] at  (M1) {$x^*-\lambda \cdot d$};
    
    \fill[black] (M2) circle[radius=3pt];
    \node[right,black] at  (M2) {$x^*+d$};

    \draw[black] (M1)--({sqrt(\r*\r-\h*\h)},\h,0)--(M2);
    
    \fill[black] (A) circle[radius=3pt];
    \node[right,black] at  (A) {${x}^*$};
    \end{tikzpicture}
% }
\caption{Despite non-uniqueness, IMGD contains strategies to a subspace with a unique equilibrium. 
For instance, when using \ref{eq:IMGD} in $\mathbb{R}^3$, agents' strategies remain equidistant from distinct Nash strategies $x^*+d$ and ${x}^*-\lambda \cdot d$ causing agent strategies to cycle around a lower-dimensional circle in $\{x\in \mathbb{R}^3: d^\intercal x = d^\intercal x^*$\} that uniquely contains the unique Nash equilibrium ${x}^*$  that is closest to the initial strategy $x^0$.} \label{fig:MultiEquil}
\end{figure}

\section{Two-Step Averaging and Decomposable Geometries}

With the powerful properties of \eqref{eq:IMGD} established, we now turn our attention to its learning-rate selection and geometric behavior.
Surprisingly, we find that with $\eta\to\infty$, \ref{eq:IMGD} can achieve time-average convergence after one iteration. 

\begin{theorem}\label{thm:2TimeAverage}
    Suppose there exists a Nash equilibrium satisfying $Gx^*+c=0$, $\bar x^2=\frac{x^0+x^1}{2}$, we have $||G\bar{x}^2 +c || \in O(1/\eta)$. I.e., we can accurately estimate the Nash equilibrium after computing a single gradient.
\end{theorem}

We prove this by decomposing the update rule into separable 2D dynamics, and base our analysis on these 2D projections.
In Section 4.1, we introduce the decomposition of \ref{eq:IMGD} into separable 2D dynamics; in Section 4.2, we find that \ref{eq:IMGD}  reduces to a planar rotation in 2D; in Section 4.3 we prove the fast convergence behavior in 2D; and in Section 4.4, we extend the result to higher dimensions and give the proof of Theorem \ref{thm:2TimeAverage}.

\subsection{Decomposition of \ref{eq:IMGD} into Separable 2D Dynamics}

First, to decompose a system of higher dimensions into separable 2-dimensional systems, we introduce the canonical form of a real skew-symmetric matrix:

\paragraph{Real Schur decomposition \cite[Theorem 7.4.1]{golu2013matrix}}
Let $G \in \mathbb{R}^{k\times k}$,
$G^\top = -G$, $k=k_1+k_2$.
There exists an orthogonal matrix $Q \in \mathbb{R}^{k\times k}$, i.e., $Q^\top Q = 1$, such that
\[
G = Q J Q^\top ,
\]
where $J$ is block diagonal of the form
\[
J =
\begin{pmatrix}
J(\omega_1) & & & \\
& \ddots & & \\
& & J(\omega_l) & \\
& & & 0
\end{pmatrix},
\quad
J(\omega_j)=
\begin{pmatrix}
0 & \omega_j \\
-\omega_j & 0
\end{pmatrix},
\qquad \omega_j \in \mathbb{R}.
\]

and 
$\lambda = \pm i\omega_j (j=1,\dots,l)
$ are eigenvalues,
zero eigenvalues corresponding to the zero block in $J$.   

The structure of J provides insight into analyzing the dynamics on separable 2 × 2 subspaces. And we find that
\ref{eq:IMGD} dynamics can decompose into separable 2D dynamics:

\begin{theorem}\label{thm:decomposition}
Suppose $G$ has $2l$ non-zero eigenvalues and a  Nash equilibrium \(x^\star\) exists,
then  the sequence $\{x^t\}_{t=1}^{\infty}$ can be decomposed
into $\ell$ separable 2-dimensional dynamics.
Formally, for
$y^t = Q^\top x^t$ where $Q$ is the orthogonal matrix in the real Schur decomposition,
\[y_j^{t+1}=y_j^t
,\quad j>2l;\] 
\begin{align*}
y_{2j-1}^{t+1}
&=
y_{2j-1}^{t}
+
\eta
\left(
\omega_j
\frac{
y_{2j}^{t}+y_{2j}^{t+1}
}{2}
+
\bigl[Q^\top c\bigr]_{2j-1}
\right),
\\
y_{2j}^{t+1}
&=
y_{2j}^{t}
+
\eta
\left(
-\omega_j
\frac{
y_{2j-1}^{t}+y_{2j-1}^{t+1}
}{2}
+
\bigl[Q^\top c\bigr]_{2j}
\right),\quad j=1,\dots,l.  
\end{align*}

I.e.\@, \ref{eq:IMGD} dynamics are equivalent to
$l$ instances of \ref{eq:IMGD} on 1-strategy zero-sum games with payoff matrices
$\begin{bmatrix}
 \omega_j \
\end{bmatrix}
_{j=1}^{\ell}$.
\end{theorem}

\begin{proof}
Multiplying $\{x^t\}_{t=1}^\infty$ by $Q^\top$ gives
\begin{align*}
& Q^\top x^{t+1} = Q^\top (x^t + G\frac{x^t + x^{t+1}}{2} + c) \\
\Leftrightarrow \ & 
 Q^\top x^{t+1} =  Q^\top x^t + J Q^\top\frac{x^t + x^{t+1}}{2} + c) \\
\Leftrightarrow \ & y^{t+1}
=
y^t
+
\eta
\left(
J\frac{y^t+y^{t+1}}{2}
+
Q^\top c
\right).
\end{align*}

Since $J$ is block diagonal, the coordinates associated with different blocks are completely decoupled. 

For every zero eigenvalue coordinate $j>2\ell$, the corresponding
diagonal entry of $J$ is zero, so
\[
y_j^{t+1}
=
y_j^t+\eta [Q^\top c]_j.
\]
The Nash equilibrium \(x^\star\) satisfies
$Gx^\star + c = 0$, we have
$c = -Gx^\star = -QJQ^\top x^\star $, so
\[
\bigl[Q^\top c\bigr]_{j} = \bigl[-J(Q^\top x^\star) \bigr]_{j} =0 \Rightarrow y_j^{t+1}
=
y_j^t,\quad \forall j>2l.
\]
For each nonzero block
\[
J(\omega_j)
=
\begin{bmatrix}
0 & \omega_j\\
-\omega_j & 0
\end{bmatrix},
\]
let
\[
z_j^t=
\begin{bmatrix}
y_{2j-1}^t\\
y_{2j}^t
\end{bmatrix},
\qquad
d_j=
\begin{bmatrix}
[Q^\top c]_{2j-1}\\
[Q^\top c]_{2j}
\end{bmatrix}.
\]
Restricting the dynamics to this block yields
\[
z_j^{t+1}
=
z_j^t
+
\eta
\left(
J(\omega_j)
\frac{z_j^t+z_j^{t+1}}{2}
+d_j
\right).
\]

Writing the two coordinates explicitly gives
\begin{align*}
y_{2j-1}^{t+1}
&=
y_{2j-1}^{t}
+
\eta
\left(
\omega_j
\frac{y_{2j}^{t}+y_{2j}^{t+1}}{2}
+
[Q^\top c]_{2j-1}
\right),\\
y_{2j}^{t+1}
&=
y_{2j}^{t}
+
\eta
\left(
-\omega_j
\frac{y_{2j-1}^{t}+y_{2j-1}^{t+1}}{2}
+
[Q^\top c]_{2j}
\right).
\end{align*}
\end{proof}

\subsection{Planar Rotations in 2D}

Now that we have established the system can be decomposed into separable 2D systems, we study the dynamics in 2D.  We begin by establishing that dynamics rotate around the Nash equilibrium with constant velocity that is determined as a function of the learning rate. 

Consider the update rule in Theorem~\ref{thm:decomposition}, in every 2D systems, with $G=
\begin{pmatrix}
0 & \omega\\
-\omega & 0
\end{pmatrix}$, $c=
\begin{pmatrix}
b_1\\
b_2
\end{pmatrix}$ and $D =\big(I - \tfrac{\eta}{2} G\big)^{-1} \eta c$, the \ref{eq:IMGD} update  is 

\begin{equation}
\label{eq:updatein2}x^{t+1} = \Phi_\eta x^t+D, \quad
\Phi_{\eta}
=
\begin{pmatrix}
1 & -\frac{\eta \omega}{2}\\
\frac{\eta \omega}{2} & 1
\end{pmatrix}^{-1}
\begin{pmatrix}
1 & \frac{\eta \omega}{2}\\
-\frac{\eta \omega}{2} & 1
\end{pmatrix}
=\frac{1}{1+\frac{\eta^2 \omega^2}{4}}
\begin{pmatrix}
1-\frac{\eta^2 \omega^2}{4} & \eta \omega\\
-\eta \omega & 1-\frac{\eta^2 \omega^2}{4}
\end{pmatrix}.
\end{equation}

\begin{theorem} \label{thm:rotation}
Suppose $\omega\neq 0$, and there is a unique Nash equilibrium $x^*=(b_2/\omega, -b_1/\omega)$. 
Then $x^{t+1}$ is obtained by rotating $x^t$ around $x^*$ by $\theta = 2\arctan{\frac{\eta \omega}{2}}$.  Formally, $x^{t+1}-x^\star
=
R(\theta)\bigl(x^t-x^\star\bigr)$ where $ R({\theta})
=
\begin{pmatrix}
\cos\theta & -\sin\theta \\
\sin\theta & \cos\theta
\end{pmatrix}.$
\end{theorem}

\begin{proof}
Define the shifted variable $z^t=x^t-x^\star$, we have
\[
z^{t+1}
= x^{t+1} - x^\star
= \Phi_\eta x^t + D - (\Phi_\eta x^\star + D)
= \Phi_\eta (x^t - x^\star)
= \Phi_\eta z^t.
\]
From (\ref{eq:updatein2}), we  get
\[
\theta = 2\arctan\!\left(\frac{\eta \omega}{2}\right),\quad \Phi_\eta=R({\theta})
=
\begin{pmatrix}
\cos\theta & -\sin\theta \\
\sin\theta & \cos\theta
\end{pmatrix},
\]
and $R(\theta)$ is a planar rotation matrix. So
\begin{align*}&
z^{t+1}=\Phi_\eta z^t=R(\theta)z^t  
\\
\Leftrightarrow\ \ &
x^{t+1}-x^\star
=
R(\theta)\bigl(x^t-x^\star\bigr)
\end{align*}
and every iterate is obtained from the previous by a planar rotation 
around the Nash equilibrium $x^\star$ by the angle $\theta$.
\end{proof}
Figure~\ref{fig:rotation} illustrates this geometric interpretation. 
Further, as $\eta \to \infty$,  the rotation approaches $\pi$ as depicted in Figure~\ref{fig:largreta}. 
Formally, $|\theta| = |2\arctan\!\left(\frac{\eta \omega}{2}\right)| \to \pi$.

\begin{figure}[ht]
\centering

\begin{subfigure}[t]{0.48\textwidth}
\centering
\begin{tikzpicture}[scale=2.2]

% axes
\draw[->] (-0.2,0) -- (1.3,0) node[right] {$x_1$};
\draw[->] (0,-0.2) -- (0,1.3) node[above] {$x_2$};

% Nash equilibrium at the origin
\fill (0,0) circle (0.4pt);
\node[below left] at (0,0) {$x^\star$};

% quarter circle trajectory
\draw[dashed] (1,0) arc (0:90:1);

% vectors
\draw[->, thick, blue] (0,0) -- (1,0)
node[above right] {$x^t$};

\draw[->, thick, red] (0,0) -- (0.87,0.5)
node[right] {$x^{t+1}$};

\draw[->, thick, red] (0,0) -- (0.5,0.87)
node[above right] {$x^{t+2}$};

\draw[->, thick, red] (0,0) -- (0,1)
node[left] {$x^{t+3}$};

% angle markings
\draw (0.3,0) arc (0:30:0.3);
\node at (0.33,0.12) {$\theta$};

\draw (0.35,0) arc (0:60:0.35);
\node at (0.18,0.33) {$2\theta$};

\draw (0.4,0) arc (0:90:0.4);
\node at (0.05,0.45) {$3\theta$};

\end{tikzpicture}

\caption{\ref{eq:IMGD} update rotates around the Nash equilibrium with constant angular velocity
$\theta=2\arctan(\eta\omega/2)$.}
\label{fig:rotation}
\end{subfigure}
\hfill
\begin{subfigure}[t]{0.48\textwidth}
\centering
\begin{tikzpicture}[scale=1.4]

% axes
\draw[->] (-1.2,0) -- (1.2,0) node[right] {$x_1$};
\draw[->] (0,-1.2) -- (0,1.2) node[above] {$x_2$};

% invariant circle
\draw[dashed, gray] (0,0) circle (1);

% points on circle
\coordinate (xk) at (0.8,0.6);
\coordinate (xk1) at (-0.78,-0.62);
\coordinate (xk2) at (0.75,0.66);

\fill (xk) circle (0.015);
\fill (xk1) circle (0.015);
\fill (xk2) circle (0.015);

\node[above right] at (xk) {$x^t$};
\node[below left] at (xk1) {$x^{t+1}$};
\node[above left] at (xk2) {$x^{t+2}$};

\draw[thick, red, ->] (xk) -- (xk1);
\draw[thick, red, ->] (xk1) -- (xk2);

% drift arc
\draw[->, dashed]
    (xk) to[bend left=25] (xk2);

% annotation in upper-right corner
\node[] at (1.95,1.05)
    {\scriptsize $x^{t+2}-x^t$
     \scriptsize $\in O(1/\eta)$};

\draw[->, dashed]
    (1.0,0.95) -- (0.78,0.67);

\fill (0,0) circle (0.4pt);
\node[below right] at (0,0) {$x^\star$};

\end{tikzpicture}

\caption{As $\eta\to\infty$, the rotation angle approaches $\pi$, $|\theta| = |2\arctan\!\left(\frac{\eta \omega}{2}\right)| \to \pi$.}
\label{fig:largreta}
\end{subfigure}

\caption{Geometric interpretation of \ref{eq:IMGD} in the two-dimensional bilinear game.}
\label{fig:twodim_geometry}
\end{figure}
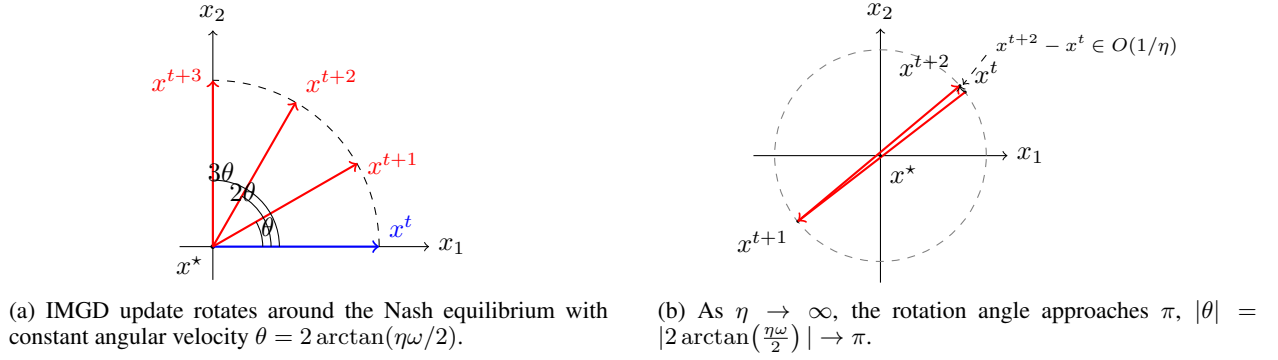

\subsection{Convergence After One Iteration in 2D}

We can directly connect the angular movement to ergodic convergence rates after a single iteration. 
Specifically, since the rotation approaches $\pi$ as $\eta\to \infty$, the dynamics approximately start from the starting location in the second iteration.  
We show that larger learning rates drive the two-step average iterate closer to the Nash equilibrium, with approximation error proportional to $\frac{1}{\eta}$:
\begin{theorem}\label{thm:converge2} Suppose there is a Nash equilibrium and let $x^*=P(x^0)$ be the Nash equilibrium closest to $x^0$. For $\bar x^{2}
=
\frac{x^0+x^1}{2}$, $\|\bar{x}^2-x^\star\|\in O\!\left(\frac{1}{\eta}\right)$.
\end{theorem}

\begin{proof}
If there are multiple equilibria, then the payoff matrix cannot be invertible, implying $\omega=0$, $G$ is a matrix of zeroes, $x^*=x^0=x^t$ for all $t$ and the result trivially holds. 

Next, consider when there is a unique Nash equilibria $x^*=(b_2/\omega, -b_1/\omega)$. 
Using the shifted variable
$z^t=x^t-x^\star$,
we get
\[
\bar{x}^2-x^\star
=
\frac{(x^0-x^\star)+(x^{1}-x^\star)}{2}
=
\frac{z^0+z^{1}}{2}.
\]
As in the proof of Theorem \ref{thm:rotation}, $z^{1}=R(\theta)z^0$ and
\[
\bar{x}^2-x^\star
=
\frac{(I+R(\theta))z^0}{2}.
\]
Taking the $l^2$-norm,
\[
\left\|\frac{z^0+z^{1}}{2}\right\|
=
\left\|
\frac{(I+R(\theta))z^0}{2}
\right\|
\le
\frac{\|I+R(\theta)\|}{2}\,\|z^0\|.
\]
Using
\[
I+R(\theta)
=
2\cos\!\left(\frac{\theta}{2}\right)
R\!\left(\frac{\theta}{2}\right),
\]
and the fact that \(R(\theta/2)\) is orthogonal, we obtain
\[
\left\|\frac{I+R(\theta)}{2}\right\|
=
\cos\!\left(\frac{\theta}{2}\right).
\]
Since \(\theta = 2\arctan(\eta\omega/2)\),
\[
\cos\!\left(\frac{\theta}{2}\right)
=
\cos\!\left(\arctan\!\left(\frac{\eta\omega}{2}\right)\right)
=
\frac{1}{\sqrt{1+\frac{\eta^2\omega^2}{4}}},
\]
which implies
\[
\left\|\frac{I+R(\theta)}{2}\right\|
=
\frac{1}{\sqrt{1+\frac{\eta^2\omega^2}{4}}}
\in O\!\left(\frac{1}{\eta}\right).
\]
Therefore, $\|\bar{x}^2-x^\star\|=\| \bar{z}^2\| \in O\!\left(\frac{1}{\eta}\right)$.
\end{proof}

\subsection{Convergence After One Iteration in Arbitrary Dimension}

Combining Theorems 4.1 and 4.3, we show \ref{eq:IMGD} obtains time-average convergence after 2 iterations in higher dimensions.

\begin{proof}[Proof of Theorem \ref{thm:2TimeAverage}.]
Let \(G = QJQ^\top\) be the real Schur decomposition of \(G\),  $y^t = Q^\top x^t$ as in Theorem \ref{thm:decomposition}, and 
$\bar y^2 = \frac{y^1+y^2}{2}$.
Since \(Q\) is orthogonal,
\[
\|G\bar x^2+c\|^2
=
\|J\bar y^2+Q^\top c\|^2
=
\sum_{j=1}^{k}
\bigl[(J\bar y^2+Q^\top c)_j\bigr]^2.
\]

By Theorem~\ref{thm:decomposition}, for $ j>2\ell$ 
\[
(Q^\top c)_j=0,
\quad y_j^{t+1}=y_j^t.
\]
Since the last \(k-2\ell\) rows of \(J\) are zero,
\[
(J\bar y^2+Q^\top c)_j=0,
\qquad j>2\ell.
\]
Therefore,
\[
\|G\bar x^2+c\|^2
=
\sum_{j=1}^{2\ell}
\bigl[(J\bar y^2+Q^\top c)_j\bigr]^2.
\]

From Theorem~\ref{thm:decomposition}, the dynamics decompose into
\(\ell\) independent two-dimensional subsystems for
\(j=1,\ldots,\ell\):
\[
\begin{aligned}
y_{2j-1}^{t+1}
&=
y_{2j-1}^{t}
+
\eta\left(
\omega_j\frac{y_{2j}^{t}+y_{2j}^{t+1}}{2}
+
(Q^\top c)_{2j-1}
\right),\\
y_{2j}^{t+1}
&=
y_{2j}^{t}
+
\eta\left(
-\omega_j\frac{y_{2j-1}^{t}+y_{2j-1}^{t+1}}{2}
+
(Q^\top c)_{2j}
\right).
\end{aligned}
\]
This is precisely the two-dimensional \ref{eq:IMGD} system considered in Theorem~\ref{thm:converge2}. Hence,
\[
\left\|
J(\omega_j){\bar y^{2(j)}}
+
\begin{pmatrix}
(Q^\top c)_{2j-1}\\
(Q^\top c)_{2j}
\end{pmatrix}
\right\|
\in
O\!\left(\frac{1}{\eta}\right).
\]

Squaring both sides yields
\[
[(J\bar y^2+Q^\top c)_{2j-1}]^2
+
[(J\bar y^2 +Q^\top c)_{2j}]^2
\in
O\!\left(\frac{1}{\eta^2}\right).
\]

Summing over all \(\ell\) blocks,
\[
\|G\bar x^2+c\|^2
\in
\sum_{j=1}^{\ell}
O\!\left(\frac{1}{\eta^2}\right)
\in
O\!\left(\frac{1}{\eta^2}\right),
\]
where the last equality follows from the fact that \(\ell\) is independent of \(\eta\). Taking square roots gives
\[
\|G\bar x^2+c\|
\in
O\!\left(\frac{1}{\eta}\right).
\]
\end{proof}

\section{Numerical Experiments}
In this section, we evaluate the performance of \ref{eq:IMGD} by comparing it with two other standard rules, Alternating Gradient Descent (AGD) and Optimistic Gradient Descent (OGD). 
In Section 5.1, we first discuss the choice of learning rate and run several experiments for \ref{eq:IMGD} as it has no limits for $\eta$. 
Then in Section 5.2, we run these three methods together for both fixed iterations and fixed time, and compare the outcomes. We don't run dedicated experiments for analyzing $\bar x^2$ since there's no baseline, as \ref{eq:IMGD} is the only known method to yield arbitrarily good approximations in a single iteration; however, the two-step time-average convergence still naturally appears in our analysis of IMGD for large learning rates.

For our benchmark methods, AGD is a classical gradient-based method as introduced in Section~\ref{sec:GradientBased}, and OGD is also a widely studied first-order method for zero-sum games, which incorporates information from the previous gradient to anticipate the rotational behavior of game dynamics:
\begin{equation}\begin{aligned}
x_1^{t+1}
&=
x_1^t
+
2\eta \nabla_{x_1}u_1(x_1^t,x_2^t)
-
\eta \nabla_{x_1}u_1(x_1^{t-1},x_2^{t-1}),
\\
x_2^{t+1}
&=
x_2^t
+
2\eta \nabla_{x_2}u_2(x_1^t,x_2^t)
-
\eta \nabla_{x_2}u_2(x_1^{t-1},x_2^{t-1}).
\end{aligned}\tag{OGD}\end{equation}

For settings in the experiments, we solve a two-player zero-sum game:
\begin{align*}
\max_{x_1 \in X_1} \; \min_{x_2 \in X_2} 
\; x_1^\top A x_2 + b_1 \cdot x_1 - b_2 \cdot x_2,
\end{align*}
where $x_1 \in X_1 \subset \mathbb{R}^{k_1}$ , $x_2 \in X_2 \subset \mathbb{R}^{k_2}$, $A \in \mathbb{R}^{k_1 \times k_2}$ and $b_1 \in \mathbb{R}^{k_1}$, 
$b_2 \in \mathbb{R}^{k_2}$ are fixed vectors. The payoff matrix $A \in \mathbb{R}^{k_1 \times k_2}$ has entries independently sampled from the uniform distribution $[-1,1]$, while the vectors $b_1 \in \mathbb{R}^{k_1}$ and $b_2 \in \mathbb{R}^{k_2}$ are independently drawn from the standard normal distribution. The problem dimensions are set as $k_1 = k_2 = 20$, giving a total dimension $k = 40$. 

All numerical experiments were conducted on a Lenovo ThinkPad L14 Gen 5 laptop equipped with an AMD Ryzen 5 PRO 7535U processor (2.90 GHz) and 32 GB of RAM, running a 64-bit Windows operating system. 

\subsection{Choice of Learning Rate}

The theoretical learning rate bounds for OGD and AGD respectively are $\eta_{\text{OGD}} \le 1/2 \|A\| $ \cite{mokhtari2020convergence} and $\eta_{\text{AGD}} < 2/\|A\| $  \cite
{bailey20261}, where  \(\|A\|\) denotes the spectral norm of \(A\). We have matrix $A \in \mathbb{R}^{k_1 \times k_2}$ sampled from the uniform distribution $[-1,1]$, so $||A||\leq k_1 $ ($k_1 = k_2$) \cite{bailey20261}.
Following the theoretical step-size scaling, we parameterize the learning rates as
\(
\eta_{\mathrm{AGD}} = c\,\frac{2}{k_1}, \quad
\eta_{\mathrm{OGD}} = c\,\frac{1}{2k_1},
\) where \(c \in \{2^{0},2^{-1},2^{-2},\ldots\}\). We perform a parameter sweep over these values, with the complete tuning results reported in Appendix~\ref{appendix:parameter_tuning}. Among the tested values, we find that \(c=2^0\) consistently yields the best empirical performance. Therefore, we set the learning rate to be 
\[
\eta_{\mathrm{AGD}}=\frac{2}{k_1}, \qquad
\eta_{\mathrm{OGD}}=\frac{1}{2k_1}.
\]

For \ref{eq:IMGD}, there is no explicit upper bound on the learning rate. To select a suitable step size, we similarly parameterize it as $\eta_\mathrm{IMGD} = \frac{c}{k_1}$
where \( c > 0 \) is a tunable scalar. 
We consider a wide range of constants $
c \in \{1, 3, 5, 10, 20, 50, 80, 10^5, 10^9, 10^{16}\}$,
 and repeat each experiment 30 times. At each iteration, we compute the time-averaged iterate and
 record the distance to the set of Nash equilibria, \( \|G\bar{x}^T +c\| \), (in log scale).

 The following figures present the convergence curves for 1000 and 10 iterations. Figure~\ref{fig:learningrate4} is highly oscillatory over the \(1000\)-iteration horizon, and Figure~\ref{fig:learningrate5} clearly reveals the underlying convergence behavior in \(10\) iterations.

\begin{figure}[ht]
\centering

\begin{subfigure}[t]{0.48\textwidth}
    \centering
    \includegraphics[width=\textwidth]{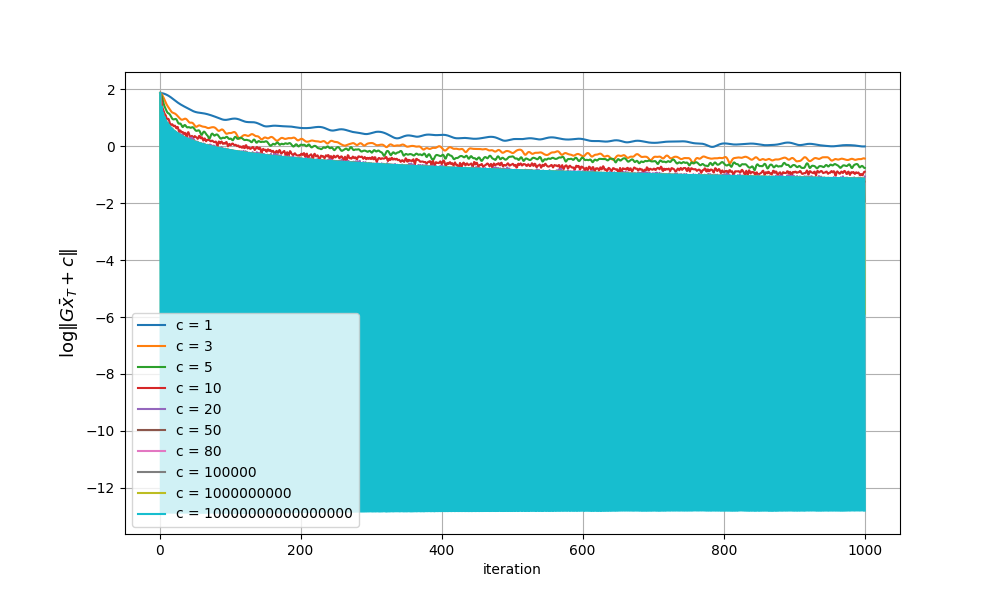}
    \caption{Time-Average Error vs Iteration for 1000 iterations}
    \label{fig:learningrate4}
\end{subfigure}
\hfill
\begin{subfigure}[t]{0.48\textwidth}
    \centering
    \includegraphics[width=\textwidth]{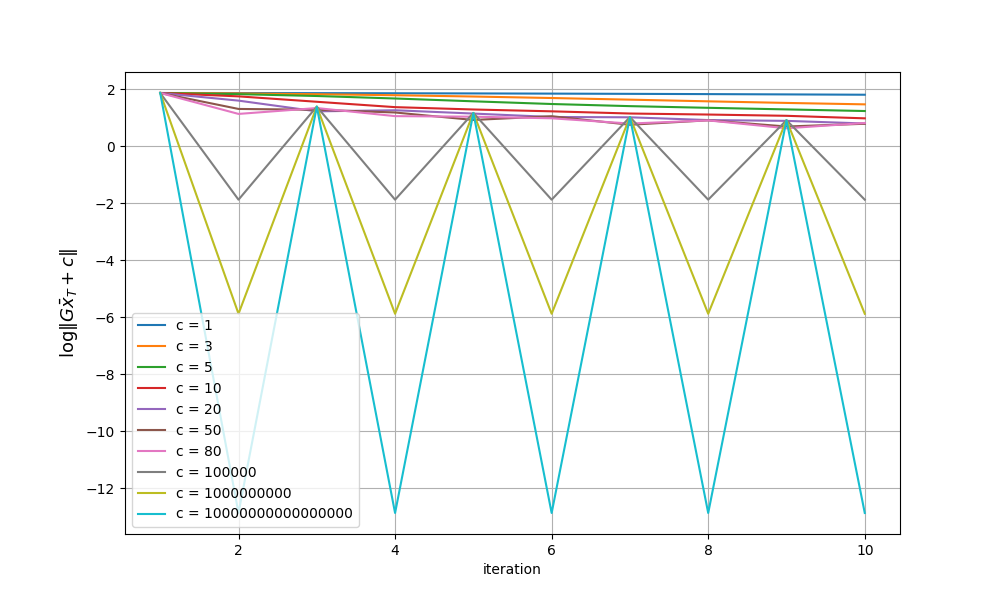}
    \caption{Time-Average Error vs Iteration for 10 iterations}
    \label{fig:learningrate5}
\end{subfigure}

\caption{Time-average error vs iteration under different iteration budgets.}
\label{fig:learningrate3}
\end{figure}

 From Figure \ref{fig:learningrate3}, it is noticeable that \ref{eq:IMGD} converges after 2 iterations for an extremely large learning rate, as indicated by Theorem \ref{thm:2TimeAverage}. However, these good estimates also make it difficult to inspect the worst-case convergence.
Thus, we examine long-term behavior by replacing each point with the maximum over future iterations to suppress oscillations and highlight convergence trends. Running the algorithm for 1000 iterations each time, the result is shown in Figure~\ref{fig:learningrate1}.

\begin{figure}[h!]
    \centering
    \includegraphics[width=0.7\textwidth]{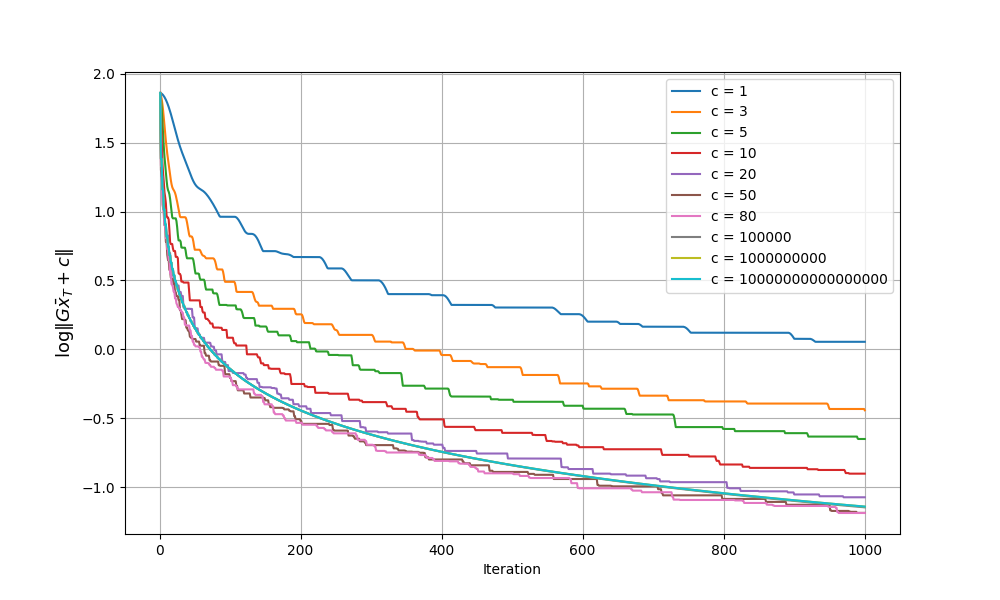}
    \caption{Smoothed Time-Average Error vs Iteration for 1000 iterations}
    \label{fig:learningrate1}
\end{figure}

As \( c \) increases, we observe that the convergence becomes faster, and very large step sizes (e.g., \( c \ge 10^5 \)) still exhibit stable behavior and achieve substantially lower error. This highlights a key advantage of \ref{eq:IMGD}: it can tolerate, and even benefit from, very large step sizes, unlike existing methods.

As indicated by Theorem \ref{thm:converge}, the improvement saturates beyond a certain point, and extremely large values of \( c \) provide diminishing returns while potentially introducing numerical sensitivity. Therefore, for a fair and representative comparison across methods, we select \( c = 10 \), as shown in Figure~\ref{fig:learningrate1}, which achieves a good balance between convergence speed, stability, and robustness.
So we  set the learning rate of \ref{eq:IMGD} to be
\[
\eta_{\mathrm{IMGD}} = \frac{10}{k_1}.
\]

\subsection{Experiments}

In the experiments, we evaluate the convergence behavior when all algorithms are run for a fixed number of iterations or time. 

\subsubsection{Fixed Iteration}
We set the initial point \( x^0 \in \mathbb{R}^n \) to be sampled uniformly from \([-1,1]^n\). All algorithms are initialized from the same point within each trial to ensure a fair comparison. Each algorithm is run for a fixed number of iterations \(T = 1000,2000,3000\). We evaluate performance using the time-averaged iterate and measure its distance to a Nash equilibrium $x^\star$. We repeat the experiment over 30 independent runs with different random initializations. 
Consistent with our theoretical analysis, we measure the distance to set of Nash equilibria with $||G\bar{x}^T+c||$ and report a 95\% confidence interval for this measured deviation from Nash.

\begin{figure*}[h!]
    \centering
    \includegraphics[width=.5\textwidth]{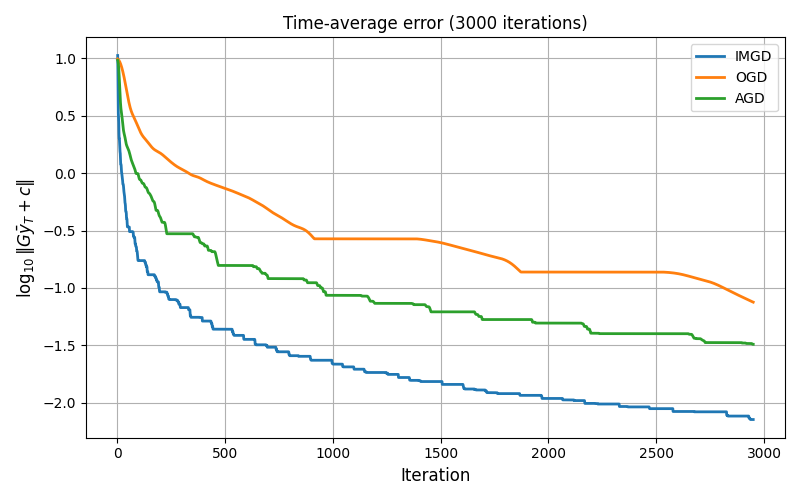}

    \caption{Smoothed time-average error vs iteration for IMGD, OGD, and AGD over 3000 iterations.}
    \label{fig:comparison_iterations}
\end{figure*}

\begin{table}[!ht]
\centering
\caption{Performance under different iteration budgets}
\label{tab:num_iters_all}

\resizebox{\textwidth}{!}{%
\begin{tabular}{lcc|cc|cc}
\toprule
& \multicolumn{2}{c}{\textbf{1000 iters}}
& \multicolumn{2}{c}{\textbf{2000 iters}}
& \multicolumn{2}{c}{\textbf{3000 iters}} \\
\cmidrule(lr){2-3} \cmidrule(lr){4-5} \cmidrule(lr){6-7}

\textbf{Algorithm}
& \textbf{Mean $||G\bar{x}^T+c||$} & \textbf{95\% CI}
& \textbf{Mean $||G\bar{x}^T+c||$} & \textbf{95\% CI}
& \textbf{Mean $||G\bar{x}^T+c||$} & \textbf{95\% CI} \\
\midrule

IMGD
& 0.0198 & [0.0159, 0.0228]
& 0.0092 & [0.0079, 0.0103]
& 0.0050 & [0.0044, 0.0059] \\
\midrule

OGD
& 0.2391 & [0.1847, 0.2990]
& 0.0864 & [0.0756, 0.1011]
& 0.0745 & [0.0569, 0.0918] \\
\midrule

AGD
& 0.0651 & [0.0510, 0.0770]
& 0.0394 & [0.0330, 0.0450]
& 0.0236 & [0.0195, 0.0274] \\
\bottomrule
\end{tabular}
}
\end{table}

\begin{table}[!ht]
\centering
\caption{Relative distance to IMGD under different iteration budgets}
\label{tab:relative_results_all}

\resizebox{\textwidth}{!}{%
\begin{tabular}{lcc|cc|cc}
\toprule
& \multicolumn{2}{c}{\textbf{1000 iters}}
& \multicolumn{2}{c}{\textbf{2000 iters}}
& \multicolumn{2}{c}{\textbf{3000 iters}} \\
\cmidrule(lr){2-3} \cmidrule(lr){4-5} \cmidrule(lr){6-7}

\textbf{Algorithm}
& \textbf{Mean Ratio} & \textbf{95\% CI}
& \textbf{Mean Ratio} & \textbf{95\% CI}
& \textbf{Mean Ratio} & \textbf{95\% CI} \\
\midrule

AGD / IMGD
& 3.1789 & [2.7927, 3.5474]
& 4.2912 & [3.9247, 4.8630]
& 4.6048 & [3.7083, 5.2138] \\
\midrule

OGD / IMGD
& 11.3698 & [9.4506, 13.8848]
& 9.0739 & [7.9007, 10.1857]
& 13.9697 & [11.3211, 16.3920] \\
\bottomrule
\end{tabular}
}
\end{table}

From Figure~\ref{fig:comparison_iterations} and Table~\ref{tab:num_iters_all}, across all iteration budgets, \ref{eq:IMGD} consistently achieves the smallest distance to equilibrium and exhibits steady improvement as the number of iterations increases. In contrast, both OGD and AGD converge more slowly and show significantly larger residual errors.

As the number of iterations increases from 1000 to 3000, the performance gap becomes more pronounced according to Table~\ref{tab:relative_results_all}. The relative distance ratios indicate that AGD is approximately $3\sim 5\times$ worse than \ref{eq:IMGD}, while OGD can be more than an order of magnitude worse, reaching over $10\times$ at 3000 iterations. These results highlight the superior stability and convergence efficiency of \ref{eq:IMGD} in this setting.

\subsubsection{Fixed Time}

We also evaluate all algorithms under a fixed wall-clock time budget. Specifically, each method is allowed to run for $t=0.05s , 0.2s, 0.5s$, and the number of iterations is determined dynamically within the time constraint. For \ref{eq:IMGD}, we include the preprocessing time required to construct the matrices used in each iteration within the total time budget. This ensures a fair comparison by accounting for both setup cost and iterative computation.

Other settings are the same as fixed iteration. Performance is reported as absolute value of the time-averaged iterate gradient, together with a 95\% confidence interval. Results are shown in Figure~\ref{fig:comparison_times}, Table ~\ref{tab:time_budget_absolute} and Table ~\ref{tab:time_budget_relative}.

\begin{figure*}[h!]
    \centering

        \includegraphics[width=0.5\textwidth]{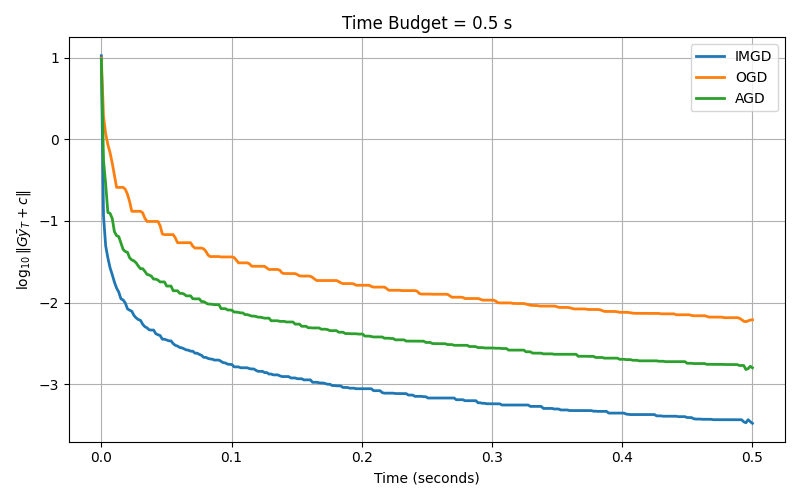}

    \caption{Smoothed time-average error vs iteration for IMGD, OGD, and AGD as a function of time spent.}
    \label{fig:comparison_times}
\end{figure*}

\begin{table}[!ht]
\centering
\caption{Absolute performance under different time budgets}
\label{tab:time_budget_absolute}

\resizebox{\textwidth}{!}{%
\begin{tabular}{lcc|cc|cc}
\toprule
& \multicolumn{2}{c}{\textbf{0.05 s}}
& \multicolumn{2}{c}{\textbf{0.2 s}}
& \multicolumn{2}{c}{\textbf{0.5 s}} \\
\cmidrule(lr){2-3} \cmidrule(lr){4-5} \cmidrule(lr){6-7}

\textbf{Algorithm}
& \textbf{Mean $||G\bar{x}^T+c||$} & \textbf{95\% CI}
& \textbf{Mean $||G\bar{x}^T+c||$} & \textbf{95\% CI}
& \textbf{Mean $||G\bar{x}^T+c||$} & \textbf{95\% CI} \\
\midrule

IMGD
& 0.000803 & [0.000571, 0.001078]
& 0.000202 & [0.000132, 0.000248]
& 0.000080 & [0.000058, 0.000113] \\
\midrule

OGD
& 0.022246 & [0.016744, 0.035808]
& 0.005349 & [0.004471, 0.006599]
& 0.002123 & [0.001748, 0.002482] \\
\midrule

AGD
& 0.003698 & [0.002714, 0.005236]
& 0.000985 & [0.000703, 0.001254]
& 0.000383 & [0.000286, 0.000487] \\
\bottomrule
\end{tabular}
}
\end{table}

\begin{table}[!ht]
\centering
\caption{Relative performance versus IMGD under different time budgets}
\label{tab:time_budget_relative}

\resizebox{\textwidth}{!}{%
\begin{tabular}{lcc|cc|cc}
\toprule
& \multicolumn{2}{c}{\textbf{0.05 s}}
& \multicolumn{2}{c}{\textbf{0.2 s}}
& \multicolumn{2}{c}{\textbf{0.5 s}} \\
\cmidrule(lr){2-3} \cmidrule(lr){4-5} \cmidrule(lr){6-7}

\textbf{Algorithm}
& \textbf{Mean Ratio} & \textbf{95\% CI}
& \textbf{Mean Ratio} & \textbf{95\% CI}
& \textbf{Mean Ratio} & \textbf{95\% CI} \\
\midrule

AGD / IMGD
& 4.7479 & [3.1832, 8.0275]
& 5.0021 & [3.1769, 7.6729]
& 4.9638 & [3.0446, 8.1713] \\
\midrule

OGD / IMGD
& 28.8080 & [16.9793, 56.1556]
& 27.0318 & [20.3495, 40.8056]
& 27.6413 & [17.3203, 39.5822] \\
\bottomrule
\end{tabular}
}
\end{table}

Across all time budgets, \ref{eq:IMGD} consistently achieves the smallest distance to equilibrium and shows steady improvement as the time budget increases. AGD and OGD converge more slowly with significantly larger residual errors. Relative performance shows that AGD is approximately $4\sim5\times$ worse than \ref{eq:IMGD}, while OGD can be $20\times$ worse, highlighting the superior stability and efficiency of \ref{eq:IMGD}.

These experiments indicate that, on the tested bilinear instances, IMGD achieves substantially smaller residuals than AGD and OGD under both fixed-iteration and fixed-time budgets.

\section{Discussion}

\subsection{Computational Costs}

The main additional computational cost of IMGD, relative to explicit methods such as AGD and OGD, is a one-time inverse computation. Specifically, to effectively implement IMGD, we precompute $\phi_\eta$ and $D$ in the update $x^{t+1}=\Phi_\eta x^t+D$.

For dense games with total dimension \(k=k_1+k_2\), this initialization has cost \(O(k^3)\) using standard dense linear algebra. After initialization, each IMGD iteration requires a dense matrix-vector multiplication and vector addition, giving per-iteration cost \(\Theta(k^2)\). This is the same per-iteration order as standard dense implementations of AGD and OGD, whose updates also require matrix-vector products involving \(A\), \(A^\top\), or \(G\).

Consequently, the inverse cost is naturally amortized over the run. If an algorithm is run for \(T\) iterations, then the iterative cost is \(\Theta(Tk^2)\), while the one-time inverse cost is \(O(k^3)\). Therefore, in the common regime where \(T\) is significantly larger than \(k\), the initialization cost is asymptotically dominated by the cumulative iteration cost. This is the relevant regime for many online optimization experiments, where the number of gradient updates is typically much larger than the problem dimension.

This comparison is also consistent with our numerical experiments. In the fixed-time experiments, the preprocessing time required by IMGD is included in the total wall-clock budget, so the reported performance already accounts for the inverse computation rather than only the cheaper post-initialization iterations. Despite this accounting, IMGD remains substantially more accurate than AGD and OGD under the tested time budgets.

Finally, the absence of an upper stability restriction on \(\eta\) reduces another source of computational overhead. Methods such as AGD and OGD require learning rates satisfying bounds depending on \(\|G\|\), and practical implementations often require additional tuning or norm estimation. IMGD is well-defined for every \(\eta>0\), so its computational cost should be viewed as a tradeoff: a one-time inverse computation in exchange for learning-rate robustness, exact orbit preservation, and strong empirical performance.

\subsection{Extensions to Normal-Form Games}

As discussed in Section~2.2, the unconstrained bilinear setting serves as a canonical model for the local behavior of learning dynamics in normal-form zero-sum games. Once the equilibrium support is identified, the dynamics are confined to the affine hull of the corresponding face of the simplex. After eliminating the affine equality constraints, the resulting dynamics are governed by an unconstrained bilinear game, so the analysis developed in this paper applies directly on the reduced subspace.

For explicit methods such as AGD and OGD, this reduction is largely conceptual: the update rule itself remains unchanged, while the analysis is performed on the active face. IMGD introduces an additional challenge. The implicit midpoint update requires computing an inverse involving the game matrix. In the constrained setting, this inverse is formed only over the coordinates corresponding to the current support of each player's strategy. Consequently, whenever the support changes, the reduced system changes as well, requiring the inverse to be recomputed.

Recomputing this inverse after every support change would limit the computational advantage obtained by preprocessing the inverse once. This suggests that IMGD is likely to be less effective during the initial exploratory phase of learning, when supports change frequently.

Instead, the structure of IMGD naturally suggests a hybrid approach. A first-stage learning algorithm, such as AGD, OGD, or another support-identification procedure, could be used to rapidly identify the equilibrium support. Once a candidate support has stabilized, IMGD could be initialized on the induced subgame to exploit its favorable geometric properties and rapidly converge while preserving bounded orbits and allowing arbitrarily large learning rates.

This perspective raises several interesting research questions. Can equilibrium support be identified reliably during learning? When is it safe to transition from an exploratory learning dynamic to IMGD? Can support changes be detected cheaply enough that inverse updates remain computationally practical? Addressing these questions may lead to learning algorithms that combine the fast exploratory behavior of explicit methods with the strong stability guarantees of structure-preserving implicit dynamics.

\section{Conclusion}

We introduced implicit midpoint gradient descent (IMGD), an implicit update rule for unconstrained bilinear two-player zero-sum games. The method is obtained by applying the implicit midpoint integrator to the continuous-time gradient dynamics, yielding a discrete-time learning rule that preserves the underlying energy structure exactly. We proved that IMGD preserves the Euclidean distance to every Nash equilibrium, has bounded orbits for every learning rate \(\eta>0\), and achieves \(O(1/T)\) ergodic convergence without an upper stability restriction on the step size.

We also showed that the dynamics decompose into independent two-dimensional rotations under the real Schur decomposition of the skew-symmetric game matrix. This geometric viewpoint explains the behavior of the method in the large-learning-rate regime and shows that a two-step average approaches the Nash residual at rate \(O(1/\eta)\). Numerical experiments further demonstrate that IMGD can substantially outperform AGD and OGD in both fixed-iteration and fixed-time comparisons on the tested bilinear instances.

Overall, the results suggest that structure-preserving discretizations provide a useful algorithmic principle for zero-sum game dynamics. The method trades a one-time inverse computation for learning-rate robustness and exact orbit preservation, motivating future work on sparse implementations, adaptive learning-rate selection, and support-identification methods for constrained normal-form games.

% \section{Conclusion}
% In this work, we propose a new implicit update rule (\ref{eq:IMGD}) for bilinear two-player zero-sum games that preserves the underlying energy structure of the continuous-time dynamics. We show that \ref{eq:IMGD} has bounded orbits, can achieve quick convergence, and is stable for arbitrary learning rates. 
% We also show that the update rule has a precise geometric interpretation: under the real Schur decomposition, IMGD decomposes into independent planar rotations whose two-step averages approach the Nash at rate \(O(1/\eta)\).
% Numerical experiments further confirm the improved stability and performance compared to standard explicit methods. Overall, our results highlight the benefits of energy-preserving discretizations and provide new insight into the role of implicit dynamics in saddle-point optimization.

\bibliographystyle{plain}
\bibliography{References}

% \appendix

\section{Proof of Theorem~\ref{thm:multiple}: Time-Average Convergence}

\begin{proof}
    First, observe that the set of Nash equilibria is affine, i.e., 
    if $x^*$ and $x^*+d$ are Nash equilibria, then $x^*+\lambda d$ is a Nash equilibrium for all $\lambda\in \mathbb{R}$: $x^*$ is a Nash equilibrium if and only if $Gx^*=-c$.
    Therefore $Gd=G(x^*+d)-Gx^*=0$ and $G(x^*+\lambda d)= Gx^*+\lambda Gd=Gx^*=c$ implying $x^*+\lambda d$ is also a Nash equilibrium.

    Since the set of Nash equilibria is affine (and therefore convex), there is a unique solution to $\min_{x\in {\cal X}^*} ||x_0-x||$ and $P(x^0)$ is well-defined. 
    
    Since $P(x^0)$ is the closest Nash equilibrium to $x^0$, the points $x^0$, $P(x^0)$, and $P(x^0)\pm d$ form two right angles and $P(x^0)-d$ and $P(x^0)+d$ are equidistant from $x^0$;
    by the Pythagorean theorem,
    % \begin{align*}
        $ ||x^0-(P(x^0)\pm d)||^2
        = ||x^0-P(x^0)||^2 + ||P(x^0)-(P(x^0)\pm d)||^2
        = ||x^0-P(x^0)||^2 + ||d||^2.$
    % \end{align*}
    
    Further, by Theorem \ref{thm:bounded}, the distance to each Nash equilibrium is invariant and therefore $||x^t-(P(x^0)+d)||^2=||x^t-(P(x^0)-d)||^2$ at every time $t\geq 0$. 
    Expanding and rearranging terms yields the equality $d^\intercal x^t=d^\intercal P(x^0)$. 
\end{proof}

\section{Learning Rate Selection for AGD and OGD}
\label{appendix:parameter_tuning}

Consider
\(
\eta_{\mathrm{AGD}} = c\,\frac{2}{k_1}, \quad
\eta_{\mathrm{OGD}} = c\,\frac{1}{2k_1},
\)
where \(c \in \{2^{0},2^{-1},2^{-2},2^{-3},2^{-4}\}\). For each algorithm, we evaluate all candidate values of \(c\) on the experimental
instances used in the paper. The corresponding convergence curves are shown in
Figures~\ref{fig:agd_tuning} and~\ref{fig:ogd_tuning}. In both cases,
\(c=1\) consistently provides the best overall empirical performance. 
\begin{figure}[h!]
    \centering
    \includegraphics[width=0.7\linewidth]{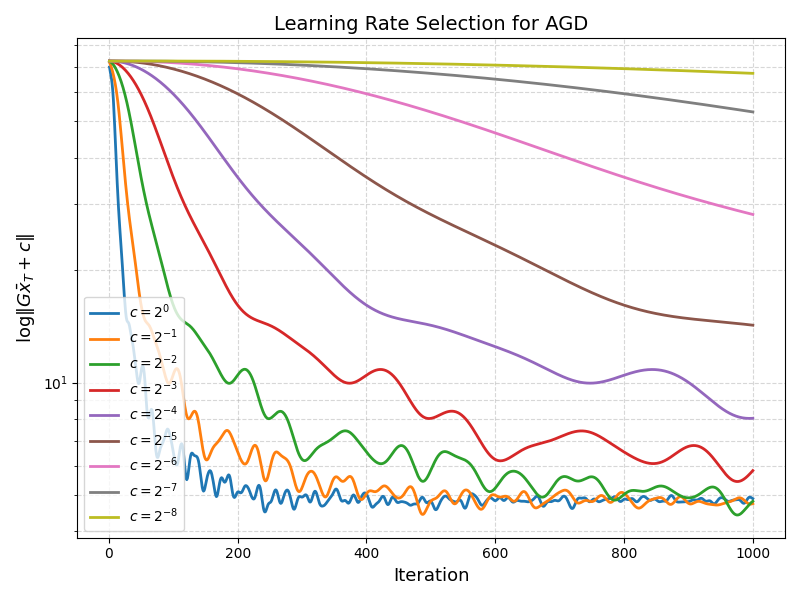}
    \caption{Parameter sweep for AltGD with
    \(\eta_{\mathrm{AGD}} = c\,2/k_1\),
    where \(c\in\{2^{0},2^{-1},2^{-2},...,2^{-8}\}\).}
    \label{fig:agd_tuning}
\end{figure}

\begin{figure}[t]
    \centering
    \includegraphics[width=0.7\linewidth]{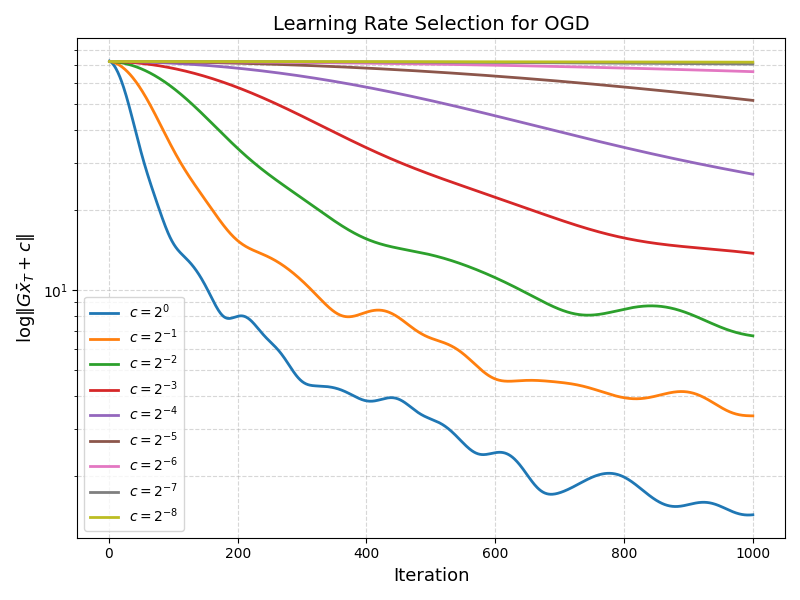}
    \caption{Parameter sweep for OGD with
    \(\eta_{\mathrm{OGD}} = c/(2k_1)\),
    where \(c\in\{2^{0},2^{-1},2^{-2},...,2^{-8}\}\).}
    \label{fig:ogd_tuning}
\end{figure}

\end{document}